\def\isarxivversion{1} %

\ifdefined\isarxivversion
\documentclass[11pt]{article}

\usepackage[numbers]{natbib}

\else
\documentclass{article}
\usepackage{neurips_2022}
\fi

\usepackage{amsmath}
\usepackage{amsthm}
\usepackage{amssymb}
\usepackage{algorithm}
\usepackage{subfig}
\usepackage{algpseudocode}
\usepackage{graphicx}
\usepackage{grffile}
\usepackage{wrapfig,epsfig}
\usepackage{url}
\usepackage{xcolor}
\usepackage{epstopdf}
\usepackage{bm}

\usepackage{bbm}
\usepackage{dsfont}

\allowdisplaybreaks

\ifdefined\isarxivversion

\usepackage{tikz}
\usepackage{hyperref}  %
\hypersetup{colorlinks=true,citecolor=blue,linkcolor=blue} %
\usetikzlibrary{arrows}
\usepackage[margin=1in]{geometry}

\else

\usepackage{microtype}
\usepackage{hyperref}
\definecolor{mydarkblue}{rgb}{0,0.08,0.45}
\hypersetup{colorlinks=true, citecolor=mydarkblue,linkcolor=mydarkblue}

\fi

\newtheorem{theorem}{Theorem}[section]
\newtheorem{lemma}[theorem]{Lemma}
\newtheorem{definition}[theorem]{Definition}

\newtheorem{fact}[theorem]{Fact}

\newcommand{\wh}{\widehat}
\newcommand{\wt}{\widetilde}
\newcommand{\ov}{\overline}

\newcommand{\R}{\mathbb{R}}

\renewcommand{\L}{\mathcal{L}}
\newcommand{\J}{\mathcal{J}}
\renewcommand{\d}{\mathrm{d}}

\newcommand{\tail}{\mathrm{tail}}

\newcommand{\T}{\mathcal{T}}

\renewcommand{\tilde}{\wt}
\renewcommand{\hat}{\wh}
\newcommand{\mmc}{\mathrm{mmc}}
\newcommand{\sym}{\mathrm{sym}} %

\DeclareMathOperator*{\E}{{\mathbb{E}}}
\DeclareMathOperator*{\var}{\mathrm{Var}}
\DeclareMathOperator*{\Z}{\mathbb{Z}}

\DeclareMathOperator*{\median}{median}

\DeclareMathOperator{\poly}{poly}

\makeatletter
\newcommand*{\RN}[1]{\expandafter\@slowromancap\romannumeral #1@}
\makeatother

\usepackage{lineno}

\begin{document}
\title{Fast Distance Oracles for Any Symmetric Norm}
\ifdefined\isarxivversion

\date{}

\author{
Yichuan Deng\thanks{\texttt{ethandeng02@gmail.com}. University of Science and Technology of China.}
\and 
Zhao Song\thanks{\texttt{zsong@adobe.com}. Adobe Research.}
\and 
Omri Weinstein\thanks{\texttt{omri@cs.columbia.edu}. The Hebrew University and Columbia University. Supported by NSF CAREER award CCF-1844887, ISF grant \#3011005535, and ERC Starting Grant \#101039914.}
\and 
Ruizhe Zhang\thanks{\texttt{ruizhe@utexas.edu}. The University of Texas at Austin.}
}

\else
\maketitle

\fi

\ifdefined\isarxivversion
\begin{titlepage}
  \maketitle
  \begin{abstract}

In the \emph{Distance Oracle} problem, the goal is to preprocess %
$n$ vectors $x_1, x_2, \cdots, x_n$ 
in a $d$-dimensional metric space $(\mathbb{X}^d, \| \cdot \|_l)$ into a cheap data structure, so that given a query vector $q \in \mathbb{X}^d$ and a subset $S\subseteq [n]$ of the input data points, all distances  $\| q - x_i \|_l$ for $x_i\in S$ can be quickly approximated (faster than the trivial $\sim d|S|$ query time). 
This primitive is a basic subroutine  in machine learning, data mining and similarity search applications. 
In the case of $\ell_p$ norms, the problem is well understood, and optimal data structures are known for most values of $p$.

Our main contribution is a fast $(1+\varepsilon)$ distance oracle for  \emph{any symmetric} norm $\|\cdot\|_l$. 
This class includes $\ell_p$ norms and Orlicz norms as special cases, as well as other norms used in practice, e.g. top-$k$ norms, max-mixture and sum-mixture of $\ell_p$ norms, small-support norms and the box-norm. 
We propose a novel data structure with  $\tilde{O}(n (d + \mathrm{mmc}(l)^2 ) )$ preprocessing time and space, and $t_q = \tilde{O}(d + |S| \cdot \mathrm{mmc}(l)^2)$ query time, for computing distances to a subset $S$ of data points, where $\mathrm{mmc}(l)$ is a complexity-measure (concentration modulus) of the symmetric norm. 
When $l = \ell_{p}$ , this runtime matches the aforementioned state-of-art oracles. %

  \end{abstract}
  \thispagestyle{empty}
\end{titlepage}

\newpage

\else

\begin{abstract}

In the \emph{Distance Oracle} problem, the goal is to preprocess %
$n$ vectors $x_1, x_2, \cdots, x_n$ 
in a $d$-dimensional metric space $(\mathbb{X}^d, \| \cdot \|_l)$ into a cheap data structure, so that given a query vector $q \in \mathbb{X}^d$ and a subset $S\subseteq [n]$ of the input data points, all distances  $\| q - x_i \|_l$ for $x_i\in S$ can be quickly approximated (faster than the trivial $\sim d|S|$ query time). 
This primitive is a basic subroutine  in machine learning, data mining and similarity search applications. 
In the case of $\ell_p$ norms, the problem is well understood, and optimal data structures are known for most values of $p$.

Our main contribution is a fast $(1+\varepsilon)$ distance oracle for  \emph{any symmetric} norm $\|\cdot\|_l$. 
This class includes $\ell_p$ norms and Orlicz norms as special cases, as well as other norms used in practice, e.g. top-$k$ norms, max-mixture and sum-mixture of $\ell_p$ norms, small-support norms and the box-norm. 
We propose a novel data structure with  $\tilde{O}(n (d + \mathrm{mmc}(l)^2 ) )$ preprocessing time and space, and $t_q = \tilde{O}(d + |S| \cdot \mathrm{mmc}(l)^2)$ query time, for computing distances to a subset $S$ of data points, where $\mathrm{mmc}(l)$ is a complexity-measure (concentration modulus) of the symmetric norm. 
When $l = \ell_{p}$ , this runtime matches the aforementioned state-of-art oracles. %

\end{abstract}

\fi

\section{Introduction}

Estimating and detecting  similarities in datasets is a fundamental problem in computer science, and a basic subroutine in most industry-scale ML applications, from  optimization \cite{cmf+20,clp20,xss21} and reinforcement learning \cite{ssx21}, to discrepancy theory \cite{sxz22} and covariance estimation \cite{v12,a19}, Kernel SVMs \cite{cho09,ssy11}, compression and clustering \cite{irw17,mmr19}, 
to mention a few. 
Such applications often need to quickly compute distances of  
online (test) points to a subset of points in the input data set (e.g., the training data) for transfer-learning and classification.  These applications have motivated the notion of \emph{distance oracles} (DO) \cite{pel00,gpp04, wp11}: In this problem,
the goal is to preprocess a dataset of  $n$ input points $X = (x_1, x_2,\ldots, x_n)$
in some $d$-dimensional metric space, into a small-space data structure which, 
given a query vector $q$ and a subset $\mathsf{S} \subseteq [n]$, can quickly estimate all the distances $\mathsf{d}(q,x_i)$ of $q$ to $\mathsf{S}$ (note that the problem of estimating a \emph{single} distance $\mathsf{d}(q,x_i)$ is not interesting in $\R^d$, as this  can be trivially done in $O(d)$ time, which is necessary to merely read the query $q$). The most well-studied case (in both theory and practice) is when the metric space is in fact a \emph{normed} space, i.e.,   the data points 
$\{x_i\}_{i\in [n]}\in \R^d$ are endowed with some predefined   
norm $\|\cdot\|$, and the goal is to quickly estimate $\|x_i-q\|$ simultaneously for all $i$, i.e., in time 
$\ll nd$ which is the trivial query time. 
Distance oracles can therefore be viewed as generalizing \emph{matrix-vector} multiplication: for the inner-product distance function $\langle x_i,q \rangle$, the query asks to approximate $X\cdot q$ in $\ll nd$ time.

For the most popular distance metrics---the Euclidean distance ($\ell_2$-norm) and Manhattan distance ($\ell_1$-norm)---classic dimension-reduction (sketching) provide very efficient distance oracles \cite{jl84,ac09,ldfu13}.  
However, in many real-world problems, these metrics  do not adequately capture similarities between data points, and a long line of work has demonstrated that more complex (possibly \emph{learnable}) metrics 
can lead to substantially better prediction and data compression \cite{dkj+07}. In particular, many works over the last decade have been dedicated to extending   various optimization problems beyond Euclidean/Manhattan distances, for example in (kernel) linear regression \cite{swy+19}, approximate nearest neighbor \cite{ann+17,ann+18}, sampling \cite{lsv18}, matrix column subset selection \cite{swz19}, and statistical queries \cite{lnrw19}.

For $\ell_p$ norms, the DO problem is well-understood \cite{bjks04},%
where the standard tool for constructing the data structure is via randomized \emph{linear sketching}  
: The basic idea is to reduce the dimension ($d$) of the data points by applying some \emph{sketching matrix} $\Phi\in \R^{m \times d}$ ($m \ll d$) to each data point $x_i$ and store the sketch $\Phi x_i\in \R^m$. For a query point $q$, linearity then allows to estimate the distance from $\Phi(q-x_i)$.  The seminal works of \cite{jl84,ams99,ccf04,tz12} developed polylogarithmic-size sketching methods for the $\ell_2$-norm, which was extended, in a long line of work, to any $\ell_p$ norm with $0<p<2$  \cite{i06,knpw11,cn20}. 
For $p>2$,  the sketch-size ($d$) becomes \emph{polynomial}, yet still sublinear in $d$ \cite{ss02, bjks04}. 

In this work, we consider general \emph{symmetric norms}, which generalize
$\ell_p$ norms. %
More formally, a norm  norm $l : \R^d \rightarrow \R$ is  symmetric if, for all $x \in \R^d$ and every $d \times d$ permutation matrix $P$, it satisfies $l(Px) = l(x)$ and also $l(|x|) = l(x)$, where $|x|$ is the coordinate-wise absolute value of $x$ 
(for a broader introduction, see \cite{b97} Chapter IV).
Important special cases of symmetric norms  are $\ell_p$ norms and  \emph{Orlicz norms} \cite{als+18}, which naturally arise in harmonic analysis and model data with sub-gaussian properties. 
Other practical examples of symmetric norms include 
top-$k$ norms, max-mix of $\ell_p$ norms, sum-mix of $\ell_p$ norms, the $k$-support norm \cite{afs12} and the box-norm \cite{mps14}. 

Several recent works have studied dimension-reduction (sketching) for  special cases of symmetric norms such as the Orlicz norm \cite{bbc+17, swy+19, als+18}, for various numerical linear algebra primitives \cite{swz19}. These sketching techniques are quite ad-hoc and are carefully tailored to the norm in question. It is therefore natural to ask whether \emph{symmetry alone} is enough to guarantee dimensionality-reduction for symmetric similarity search, in other words  

\begin{center}
    {\it Is there an efficient $(1+\epsilon)$-distance oracle for general symmetric norms? }
\end{center}
By ``efficient'', we mean 
small space and preprocessing time (ideally $\sim nd$),  
fast query time ($\ll n|S|$ for a query $(q,S\subseteq[n])$) and ideally supporting dynamic updates to $x_i$'s in $\tilde{O}(d)$ time. 
Indeed, most ML applications involve rapidly-changing \emph{dynamic datasets}, and it is becoming increasingly clear that static data structures do not adequately capture the requirements of real-world applications \cite{jkdg08,cmf+20,clp20}. As such, it is desirable to design a \emph{dynamic distance oracle} which has both small update time ($t_u$) for adding/removing a point $x_i \in \R^d$, and small query time ($t_q$) for distance estimation. We remark that most known DOs are dynamic by nature (as they rely on linear sketching techniques), but for general metrics (e.g., graph shortest-path or the $\ell_\infty$ norm) this is much less obvious, and indeed \emph{linearity} of encoding/decoding will be a key challenge in our data structure (see next section).  
The problem is formally defined as follows:
\begin{definition}[Symmetric-Norm Distance Oracles]\label{def:main_problem}
Let $\| \cdot \|_{\sym}$ denote the symmetric norm. 
We need to design a data structure that efficiently supports any sequence of the following operations:
\begin{itemize}
    \item \textsc{Init}$(\{ x_1, x_2, \cdots, x_n \} \subset \R^d, \epsilon \in (0,1), \delta \in (0,1))$. The data structure takes $n$ data points $\{x_1, x_2, \dots, x_n\}\subset \R^d$, an accuracy parameter $\epsilon$ and a failure probability $\delta$ as input. 
    \item \textsc{UpdateX}$(z \in \R^d, i \in [n])$. Update the data structure with the $i$-th new data point $z$. 
    \item \textsc{EstPair}$(i,j  \in [n])$ Outputs a number $\mathrm{pair}$ such that $(1-\epsilon) \| x_i - x_j \|_{\sym} \leq \mathrm{pair} \leq (1+\epsilon) \cdot \| x_i - x_j \|_{\sym} $ with probability at least $1 -\delta$.
    \item \textsc{Query}$(q \in \R^d)$. Outputs a vector $\mathrm{dist} \in \R^n$ such that $\forall i \in[n], (1-\epsilon)\|q-x_{i}\|_{\sym} \leq \mathrm{dst}_{i} \leq(1+\epsilon)\|q-x_{i}\|_{\sym}$. with probability at least $1 -\delta$. 
\end{itemize}
where $\| x \|_{\sym}$ is the symmetric norm of vector $x$.
\end{definition}

This problem can be viewed as an \emph{online} version of the (approximate) \emph{closest-pair} problem \cite{v12}, which asks to find the closest pair of points among an \emph{offline batch} of data points $ X = x_1,\ldots, x_n \in \R^d$, or equivalently, the smallest entry of the \emph{covariance matrix} $X X^\top$. One major (theoretical) advantage of the offline case is that it enables the use of fast matrix-multiplication (FMM) to speed-up the computation of the covariance matrix \cite{ v12, awy14, awy18, a19} (i.e., sub-linear amortized per query). By contrast, in the online setting such speedups are conjectured to be impossible \cite{hkns15,lw17}.

\paragraph{Notations.}
For any positive integer $n$, we use $[n]$ to denote $\{1,2,\cdots,n\}$. For any function $f$, we use $\wt{O}(f)$ to denote $f \cdot \poly (\log f)$. We use $\Pr[]$ to denote probability. We use $\E[]$ to denote expectation. We use both $l(\cdot)$ and $\| \cdot \|_{\sym}$ to denote the symmetric norm. We use $\| \cdot \|_2$ to denote the entry-wise $\ell_2$ norm. We define a tail notation which is very standard in sparse recover/compressed sensing literature. For any given vector $x \in \R^d$ and an integer $k$, we use $x_{\ov{[k]}}$ or $x_{\tail(k)}$ to denote the vector that without (zeroing out) top-$k$ largest entries (in absolute). For a vector $x$, we use $x^\top$ to denote the transpose of $x$. For a matrix $A$, we use $A^\top$ to denote the transpose of $A$. We use ${\bf 1}_n$ denote a length-$n$ vector where every entry $1$. We use ${\bf 0}_n$ denote a length-$n$ vector where every entry is $0$.

\subsection{Our Results}

Two important complexity measures of (symmetric) norms,  which capture their ``intrinsic dimensionality", are the \emph{concentration modulus} ($\mathrm{mc}$) and \emph{maximum modulus} ($\mathrm{mmc}$) parameters. We now define these quantities along the lines of \cite{bbc+17,swy+19}.

\begin{definition}[Modulus of concentration ($\mathrm{mc}$)]%
    \label{def:mc}
    Let $X\in\mathbb{R}^n$ be uniformly distributed on $S^{n-1}$, the $l_2$ unit sphere. The median of a symmetric norm $l$ is the (unique) value $M_l$ such that $\Pr{[l(X)\ge M_l]}\ge1/2$ and  $\Pr{[l(X)\le M_l]}\ge1/2$. Similarly, $\mathfrak{b}_l$ denotes the maximum value of $l(x)$ over $x\in S^{n-1}$. We call the ratio
    \begin{align*}
        \mathrm{mc}(l):=\mathfrak{b}_l/M_l
    \end{align*}
    the modulus of concentration of the norm $l$.
\end{definition}

For every $k\in [n]$, the norm $l$ induces a norm $l^{(k)}$ on $\mathbb{R}^k$ by setting
\begin{align*}
    l^{(k)}((x_1,x_2,\dots,x_k)):=l((x_1,x_2,\dots,x_k,\dots,0,\dots,0))
\end{align*}

\begin{definition}[$\mathrm{mmc}$]%
    \label{def:mmc}
    Define the maximum modulus of concentration of the norm $l$ as 
    \begin{align*}
        \mmc(l):=\max_{k\in [n]}\mathrm{mc}(l^{(k)})=\max_{k \in [n]}\frac{\mathfrak{b}_{l^{(k)}}}{M_{l^{(k)}}}
    \end{align*}
\end{definition}
Next, we present a few examples (in Table~\ref{tab:mmc_examples}) for different norm $l$'s $\mmc(l)$.

\begin{table}[htbp]
    \centering
    \begin{tabular}{|l|l|} \hline
        {\bf Norm} $l$ & $\mmc(l)$  \\ \hline
        $\ell_p(p \le 2)$ & $\Theta(1)$ \\ \hline  
        $\ell_p(p > 2)$ & $\Theta(d^{1/2 - 1/p})$ \\ \hline 
        top-$k$ norms & $\wt{\Theta}(\sqrt{d/k})$ \\ \hline 
        $k$-support norms and the box-norm & $O(\log d)$ \\ \hline 
        max-mix and sum-mix of $\ell_1$ and $\ell_2$ & $O(1)$ \\ \hline 
        Orlicz norm $\|\cdot\|_G$ & $O(\sqrt{C_G\log d})$ \\ \hline 
    \end{tabular}
    \caption{Examples of $\mmc(l)$, where max-mix of $\ell_1$ and $\ell_2$ is defined as $\max\{\|x\|_2, c\|x\|_1\}$ for a real number $c$, sum-mix of $\ell_1$ and $\ell_2$ is defined as $\|x\|_2 + c\|x\|_1$ for a real number $c$, and $C_G$ of Orlicz norm is defined as the number that for all $0 < x < y$, $G(y)/G(x) \le C_G(y/x)^2$.%
    }
    \label{tab:mmc_examples}
\end{table}

We are now ready to state our main result:

\begin{theorem}[Main result, informal version of Theorem~\ref{thm:main_formal}]\label{thm:main_informal}
Let $\|\cdot\|_l$ be any symmetric norm on $\R^d$.
There is a data structure  for the online 
symmetric-norm Distance Oracle problem (Definition ~\ref{def:main_problem}), which uses $ n (d + \mmc(l)^2 ) \cdot \poly(1/\epsilon, \log(n d/\delta))$ space, 
supporting the following operations: 
\begin{itemize}
    \item \textsc{Init}$( \{x_1, x_2, \dots, x_n\}\subset \R^d, \epsilon \in (0,1), \delta \in (0,1))$: Given $n$ data points $\{x_1, x_2, \dots, x_n\}\subset \R^d$, an accuracy parameter $\epsilon$ and a failure probability $\delta$ as input, the data structure preprocesses in time $ n (d + \mmc(l)^2 ) \cdot \poly(1/\epsilon, \log(n d /\delta))$. Note that $\mmc()$ is defined as Definition~\ref{def:mmc}.
    
    \item \textsc{UpdateX}$(z \in \R^{d}, i \in [n])$: Given an update vector $z \in \R^d$ and index $i \in [n]$, the data data structure receives $z$ and $i$ as inputs, and updates the $i$-th data point $x_i\leftarrow z$, in $d \cdot \poly(1/\epsilon, \log(n d/\delta))$ time.

    \item \textsc{Query}$(q \in \R^d , \mathsf{S} \subseteq [n])$: Given a query point $q \in \R^d$ 
    and a subset of the input points $\mathsf{S} \subseteq [n]$, the \textsc{Query} operation outputs a $\epsilon$- multiplicative approximation to each  distance from $q$ to points in $\mathsf{S}$, in time 
    \begin{align*}
     (d + |\mathsf{S}| \cdot \mmc(l)^2 ) \cdot \poly(1/\epsilon, \log(nd/\delta))
    \end{align*}
    i.e. it provides a set of estimates $\{\mathrm{dst}_i\}_{i \in \mathsf{S} }$ such that:
    \begin{align*}
       \forall i \in \mathsf{S}, (1-\epsilon)\|q-x_{i}\|_{l} \leq \mathrm{dst}_{i} \leq(1+\epsilon)\|q-x_{i}\|_{l} 
    \end{align*}
     with probability at least $1 -\delta$. 
    \item \textsc{EstPair}$(i,j  \in [n])$ Given indices $i, j \in [n]$, the \textsc{EstPair} operation takes $i$ and $j$ as input and approximately estimates the symmetric norm distances from $i$-th to the $j$-th point $x_i, x_j \in \R^d$ in time 
    $
     \mmc(l)^2 \cdot \poly(1/\epsilon, \log(n d / \delta))
     $
    i.e. it provides a estimated distance $\mathrm{pair}$ such that:
    \begin{align*}
       (1-\epsilon)\|x_i-x_j\|_{l} \leq \mathrm{pair} \leq(1+\epsilon)\|x_i-x_j\|_{l} 
    \end{align*}
     with probability at least $1 -\delta$.
\end{itemize}
\end{theorem}

\paragraph{Roadmap.}
In Section~\ref{sec:tech_overview}, we give an overview of techniques we mainly use in the work. Section~\ref{sec:sparse_rc_ds} gives an introduction of the Sparse Recovery Data we use for sketching. Section~\ref{sec:running_time} analyze the running time and space of our data structure with their proofs, respectively. Section~\ref{sec:correctness} show the correctness of our data structure and give its proof. Finally in Section~\ref{sec:conclusion} we conclude our work. %
\section{Technique Overview}
\label{sec:tech_overview}

Our distance oracle follows the ``sketch-and-decode'' approach, which was  extensively used in many other sublinear-time %
compressed sensing and sparse recovery problems \cite{pri11,hikp12a,lnnt16,ns19,sswz22}. 
The main idea is to compress the data points into smaller dimension by computing, for each data point $x_i\in \R^d$, a (randomized) linear sketch $\Phi\cdot x_i\in \R^{d'}$ with $d'\ll d$ at preprocessing time, where $\Phi x_i$ is an unbiased estimator of $\|x_i\|$. 
At query time, given a query point $q\in \R^d$, we analogously compute its sketch $\Phi q$.  By \emph{linearity} of $\Phi$, the distance between $q$ and $x_i$ (i.e., $\ell(q-x_i)$) can be trivially decoded from the sketch \emph{difference} $\Phi q-\Phi x_i$. As we shall see, this simple virtue of linearity is less obvious to retain when dealing with general symmetric norms.

\paragraph{Layer approximation}
Our algorithm uses the \emph{layer approximation} method proposed by Indyk and Woodruff in \cite{iw05} and generalized to symmetric norms in \cite{bbc+17,swy+19}.
Since symmetric norms are invariant under reordering of the coordinates, the main idea in \cite{iw05} is to construct a ``layer vector'' as follows: for a vector $v\in \R^d$, round (the absolute value of) each coordinate to the nearest power $\alpha^i$ for some fixed $\alpha\in \R$ and $i\in \mathbb{N}$, and then sort the coordinates in an increasing order. %
This ensures that the $i$-th layer contains all the coordinates $j\in [d]$ satisfying: %
    $\alpha^{i - 1} < |v_j| \le \alpha^i$.
In particular, the layer vector of $v$ has the form 
\begin{align*}
    \L(v) := (\underbrace{\alpha^1, \dots, \alpha^1}_{b_1~\text{times}},~ \underbrace{\alpha^2,\dots,\alpha^2}_{b_2~\text{times}},~ \cdots, \underbrace{\alpha^P,\dots,\alpha^P}_{b_P~\text{times}},~ 0, \dots, 0) \in \R^d,
\end{align*}
where $b_i$ is the number of coordinates in layer-$i$. More importantly, since the norm is symmetric, the layer vector $\L(v)$ has a succinct representation: $(b_i)_{i\in [P]}$. 

Then, it suffices to estimate $b_i$ for each $i\in [P]$, where the Indyk-Woodruff sketching technique can be used to approximate the vector. At the $i$-th layer, each coordinate of $v$ is sampled with probability $P/b_i$, and then the algorithm identifies the \emph{$\ell_2$-heavy-hitters} of the sampled vector. \cite{bbc+17} gave a criterion for identifying the important layers, whose heavy-hitter coordinates in the corresponding \emph{sampled} vectors, is enough to recover the entire symmetric norm $\| v \|_{\sym}$. 

Unfortunately, this technique for norm estimation does not readily translate to estimating \emph{distances} efficiently:  
\begin{itemize}
    \item {\it Too many layers: }In previous works, each data point $x_i$ is sub-sampled \emph{independently} in  $R$ layers, i.e, generates $R$ subsets of coordinates $S_i^1,\dots, S_i^R\subset [d]$. The sketch of the query point $S(q)$ then needs to be compared to each $S(x_i)$ in every  layer. Since there are $R = \Omega(n)$ layers in \cite{bbc+17, iw05} of size $\Omega(d)$ across all all data points,  the total time complexity will be at least $\Omega(n d)$, which is the trivial query time. 
    
    \item {\it Non-linearity: } %
    The aforementioned sketching algorithms \cite{bbc+17,iw05} involve nonlinear operations, and thus cannot be directly used for distance estimation.
    
    \item {\it Slow decoding: } The aforementioned sketches take linear time to decode the distance from the sketch, %
    which is too slow for our application.
\end{itemize}
\vspace{-1mm}
To overcome these challenges, we use the following ideas: 
\vspace{-2mm}
\paragraph{Technique I: shared randomness} To reduce the number of layers, we let all the data points use the same set of layers. That is, in the initialization of our algorithm, we independently sample $R$ subsets $S^1,\dots, S^R$ with different probabilities. Then, for each data point $x_i$, we consider $(x_i)_{S^j}$ as the sub-sample for the $j$-th layer, and perform sketch on it. Hence, the number of different layer sets is reduced from $n t$ to $t$. For a query point $q$, we just need to compute the sketches for $(q)_{S^1},\dots, (q)_{S^R}$. And the distance between $q$ and $x_i$ can be decoded from $\{\Phi (x_i)_{S^j} - \Phi (q)_{S^j}\}_{j\in [R]}$. We also prove that the shared randomness in all data points will not affect the correctness of layer approximation.   
\vspace{-2mm}

\paragraph{Technique II: linearization}
We choose a different sketching method called \textsc{BatchHeavyHitter} (see Theorem~\ref{thm:batch_heavy_hitter} for details) to generate and maintain the \emph{linear sketches}, which allows us to decode the distance from sketch difference. 
\vspace{-2mm}
\paragraph{Technical III: locate-and-verify decoding} We design a \emph{locate-and-verify} style decoding method to recover distance from sketch. In our data structure, we not only store the sketch of each sub-sample vector, but also the vector itself. Then, in decoding a sketch, we can first apply the efficient sparse-recovery algorithm to identify the position of heavy-hitters. Next, we directly check the entries at those positions to verify that they are indeed ``heavy'' (comparing the values with some threshold), and drop the non-heavy indices. This verification step is a significant difference from the typical sparse recovery approaches, which employ complex (and time-consuming) subroutines to \emph{reconstruct
the values} of the heavy-hitter coordinates. Instead, our simple verification procedure eliminates all the \emph{false-positive} heavy-hitters, therefore dramatically reducing the running time of the second step, which can now be performed directly by reading-off the values from the memory.

With these three techniques, we obtain our sublinear-time distance estimation algorithm. Our data structure first generate a bunch of randomly selected subsets of coordinates as the layer sets. %
Then, for each data point, we run the \textsc{BatchHeavyHitter} procedure to sketch the sub-sample vector in each layer\footnote{The total sketch size of each data point is $\mmc(l)^2\cdot \poly\log d$. In the $\ell_p$-norm case with $p>2$, $\mmc(l)=d^{1/2-1/p}$ (see Table~\ref{tab:mmc_examples}) and our sketch size is $\wt{O}(d^{1-2/p})$, matching the lower bound of $\ell_p$-sketching. When $p\in (0,2]$, $\mmc(l)=\Theta(1)$ and our sketch size is $\wt{O}(1)$, which is also optimal.}. In the query phase, we call the \textsc{Decode} procedure of \textsc{BatchHeavyHitter} for the sketch differences between the query point $q$ and each data point $x_i$, and obtain the heavy hitters of each layer. We then select some ``important layers'' and use them to approximately recover the layer vector $\L(q-x_i)$, which gives the estimated distance $\|q-x_i\|_{\sym}$. 

Finally, we summarize the time and space costs of our data structure. Let $\epsilon$ be the precision parameter and $\delta$ be the failure probability parameter. Our data structure achieves $\wt{O}(n (d+ \mmc(l)^2) )$-time for initialization 
, $\wt{O}(d)$-time per data point update, 
and $\wt{O}(d + n \cdot \mmc(l)^2) )$-time per query. As for space cost, our data structure uses the space of $\wt{O}(n(d + \mmc(l)^2))$ in total. Note that $\mmc$ is defined as Defnition~\ref{def:mmc}.  %

\section{Sparse Recovery Data Structure}
\label{sec:sparse_rc_ds}

We design a data structure named \textsc{BatchHeavyHitter} to generate sketches and manage them. In our design, it is a ``linear sketch" data structure, and providing the following functions:
\begin{itemize}
    \item \textsc{Init}{$(\epsilon \in (0,0.1), n, d)$}. Create a set of Random Hash functions and all the $n$ copies of sketches share the same hash functions. This step takes ${\cal T}_{\mathrm{init}}(\epsilon,n,d)$ time.
    \item \textsc{Encode}{$(i \in [n], z \in \R^d, d)$}. This step takes ${\cal T}_{\mathrm{encode}}(d)$ encodes $z$ into $i$-th sketched location and store a size ${\cal S}_{\mathrm{space}}$ linear sketch. 
    \item \textsc{EncodeSingle}{$(i \in [n], j\in [d], z \in \R, d)$}. This step takes ${\cal T}_{\mathrm{encodesingle}}(d)$ updates one sparse vector $e_j z \in \R^d$ into $i$-th sketched location. 
    \item \textsc{Subtract}$(i,j,l \in [n])$. Update the sketch at $i$-th location by $j$-th sketch minus $l$-th sketch.
    \item \textsc{Decode}{$(i\in [n], \epsilon \in (0,0.1), d)$}. This step takes ${\cal T}_{\mathrm{decode}}(\epsilon,d)$ such that it returns a set $L \subseteq [d]$ of size $|L| = O(\epsilon^{-2})$ containing all $\epsilon$-heavy hitters $i\in [n]$ under $\ell_p$. Here we say $i$ is an $\epsilon$-heavy hitter under $\ell_2$ if $|x_i| \geq \epsilon \cdot \| x_{ \ov{ [\epsilon^{-2}] } } \|_2$ where $x_{\ov{[k]}}$ denotes the vector $x$ with the largest $k$ entries (in absolute value) set to zero. Note that the number of heavy hitters never exceeds $2/\epsilon^2$.
\end{itemize}

With this data structure, we are able to generate the sketches for each point, and subtract each other with its function. And one can get the output of heavy hitters of each sketch stored in it with \textsc{Decode} function. More details are deferred to Section~\ref{subsec:sparse_rc_tool}.

\section{Running Time and Space of Our Algorithm}
\label{sec:running_time}

\begin{algorithm}[!ht]\caption{Data structure for symmetric norm estimation: members, init, informal version of Algorithm~\ref{alg:main_members} and Algorithm~\ref{alg:main_init}}\label{alg:main_members_init_informal}
    \begin{algorithmic}[1]
    \State {\bf data structure} \textsc{DistanceOnSymmetricNorm} \Comment{Theorem~\ref{thm:main_formal}}
    \State {\bf members}  %
        \State \hspace{4mm} $ \{x_i \}_{i=1}^n \in \mathbb{R}^d$ 
      
        \State \hspace{4mm} %
        $\{S_{r,l,u}\}_{r\in[R], l\in[L], u\in [U]}$ \Comment{A list of the \textsc{BatchHeavyHitter}} 
        \State \hspace{4mm} $\{H_{r,l,u}
        \}_{r\in[R], l\in[L], u\in [U]} \subset [d]\times \R $ %
    \State {\bf end members}
    \State 
    
    \State {\bf public:}
     \Procedure{Init}{$ \{ x_1, \cdots, x_n\} \subset \R^d,  \delta , \epsilon $} \Comment{Lemma~\ref{lem:init_time}}
        \State Initialize the sparse-recovery data structure $\{S_{r,l,u}\}$ 
        \State Create $\{ \mathrm{bmap}_{r,l,u} \}$ shared by all $i \in [n]$ \Comment{$\{ \mathrm{bmap}_{r,l,u} \}$ is list of a layer set map}
        \For{$ i \in [n], j \in [d] $}
            \If{$\mathrm{bmap}_{r,u,l}[j] = 1$}
                \State Sample $x_{i,j}$ into each subvector $\ov{x}_{r,u,l,i}$
            \EndIf
        \EndFor
        \State Encode $\{x_{r,u,l,i}\}_{i\in [n]}$ into $\{S_{r,l,u}\}$%
     \EndProcedure
    \State {\bf end data structure}
    \end{algorithmic}
\end{algorithm}

\begin{algorithm}[!ht]\caption{Data structure for symmetric norm estimation: query, informal version of Algorithm~\ref{alg:main_query}}\label{alg:main_query_informal}
    \begin{algorithmic}[1]
    \State {\bf data structure} \textsc{DistanceOnSymmetricNorm}
    \State {\bf public:}
    \Procedure{Query}{$q \in \mathbb{R}^d$} \Comment{Lemma~\ref{lem:query_time}, \ref{lem:query_correct}}
        \State Encode $q$ into $\{S_{r,l,u}\}$
         \For{$i \in [n]$}
            \State Subtract the sketch of $x_i$ and $q$, get the estimated heavy-hitters of $q - x_i$
            \State Decode the sketch and store returned estimation of heavy-hitters into $H_{r,l,u}$
            \For{$u \in [U]$}
                \If{$H_{r,l,u}$ provide correct indices of heavy hitters}
                    \State Select it as good set and store it in $H_{r,l}$
                \EndIf
            \EndFor
            \State Generate estimated layer sizes $\{c^i_k\}_{k \in [P]}$
            \State $\mathrm{dst}_i \gets \textsc{LayerVetcorApprox}(\alpha, c^i_1, c^i_2,\dots,c^i_P, d)$
            \State Reset $\{H_{r,l,u}\}$ for next distance
         \EndFor
         \State \Return $\{\mathrm{dst}_i\}_{i \in [n]}$
    \EndProcedure
    \State {\bf end data structure}
    \end{algorithmic}
\end{algorithm}

We first analyze the running time of different procedures of our data structure \textsc{DistanceOnSymmetricNorm}. %
See Algorithm~\ref{alg:main_init}, with the linear sketch technique, we spend the time of $\wt{O}(n (d+ \mmc(l)^2))$ for preprocessing and generate the sketches stored in the data structure. When it comes a data update (Algorithm~\ref{alg:main_update}), we spend $\wt{O}(d)$ to update the sketch. And when a query comes (Algorithm~\ref{alg:main_query}), we spend $\wt{O}(d + n \cdot \mmc(l)^2))$ to get the output distance estimation. The lemmas of running time and their proof are shown below in this section.

\begin{lemma}[\textsc{Init} time, informal]
\label{lem:init_time}
    Given data points $\{x_1,x_2,\dots,x_n\} \subset \R^d$, an accuracy parameter $\epsilon \in (0,1)$, and a failure probability $\delta \in (0,1)$ as input, the procedure \textsc{init} (Algorithm~\ref{alg:main_init}) runs in time 
    \begin{align*}
    O( n (d + \mmc(l)^2) \cdot \poly( 1/\epsilon, \log(n d/\delta) ) ) . 
   \end{align*}
\end{lemma}
\begin{proof}
The \textsc{Init} time includes these parts:
\begin{itemize}
    \item Line~\ref{line:sketch_init} takes $O(R L U \cdot {\cal T}_{\mathrm{init}}(\sqrt{\beta},n,d))$ to initialize sketches  
    \item Line~\ref{line:bmap_init_1} to Line~\ref{line:bmap_init_2} takes $O(R U d L)$ to generate the \rm{bmap};
    \item Line~\ref{line:sketch_init_encode} takes $O(n d R U L \cdot \T_{\mathrm{encodesingle}}(d))$ to generate sketches %
\end{itemize}

By Theorems~\ref{thm:batch_heavy_hitter}, we have
\begin{itemize}
    \item $\T_{\mathrm{init}}(\sqrt{\beta}, n, d)) = n \cdot O(\beta^{-1} \log^2 d) = O(n \cdot \mmc(l)^2 \log^7(d) \epsilon^{-5} )$,
    \item  $\T_{\mathrm{encodesingle}}(d) = O(\log^2(d))$.
\end{itemize}

Adding them together we got the time of
\begin{align*}
        & ~ O(R L U {\cal T}_{\mathrm{init}}(\sqrt{\beta},n,d)) + O(R U d L) + O(n d R U L \cdot \T_{\mathrm{encodesingle}}(d)) \\
    =   & ~ O(R L U({\cal T}_{\mathrm{init}}(\sqrt{\beta},n,d) + n d \cdot  \T_{\mathrm{encodesingle}}(d)))\\
    =   & ~ O(\epsilon^{-4} \log(1/\delta) \log^4(d) \cdot \log(d) \cdot \log(d^2/\delta \cdot \log(n d)) (n \cdot \mmc(l)^2 \log^7(d) \epsilon^{-5} + n d \log^2(d)))\\
    =   & ~ O( n (d + \mmc(l)^2) \cdot \poly( 1/\epsilon, \log(n d/\delta) ) ),
\end{align*}
where the first step follows from merging the terms, the second step follows from the definition of $R,L,U,\T_{\mathrm{encodesingle}}(d),{\cal T}_{\mathrm{init}}$, the third step follows from merging the terms.

Thus, we complete the proof.
\end{proof}

\begin{lemma}[\textsc{Update} time, informal]\label{lem:update_time}
    Given a new data point $z \in \R^d$, and an index $i$ where it should replace the original data point $x_i \in \R^d$. The procedure \textsc{Update} (Algorithm~\ref{alg:main_update}) runs in time 
    \begin{align*}
    O(d  \cdot \poly( 1/\epsilon , \log (nd/\delta) )
    \end{align*}
\end{lemma}

\begin{proof} 
The \textsc{Update} operation calls \textsc{BatchHeavyHitter.Encode} for $RLU$ times, so it has the time of 
\begin{align*}
         O(R L U \cdot \T_{\mathrm{encode}}(d)) 
    =   & ~ O(\epsilon^{-4} \log(1/\delta) \log^4(d) \cdot \log(d) \cdot \log(d^2/\delta) \cdot d \log^2(d) \cdot \log(n d))\\
    =   & ~ O(\epsilon^{-4}d \log^9(n d/\delta))
\end{align*}
where the first step follows from the definition of $R,L,U,\T_{\mathrm{encode}}(d)$, the second step follows from
\begin{align*}
        & ~ \log(1/\delta)\log^4(d)  \log(d) \log(d^2/\delta) \log^2 (d) \log(n d) \\
    =   & ~ (\log(1/\delta)) (\log^7 d) (2\log d + \log(1/\delta)) \log(n d) \\
    =   & ~ O(\log^9(n d/\delta)).
\end{align*}
Thus, we complete the proof.
\end{proof}

Here, we present a \textsc{Query} for outputting all the $n$ distances. In Section~\ref{sec:more_detail_time}%
, we provide a more general version, which can take any input set $\mathsf{S} \subseteq [n]$, and output distance for only them in a shorter time that proportional to $|\mathsf{S}|$.
\begin{lemma}[\textsc{Query} time, informal]\label{lem:query_time}

Given a query point $q \in \R^d$, the procedure \textsc{Query} (Algorithm~\ref{alg:main_query}) runs in time 
\begin{align*}
    O(  (d + n \cdot \mmc(l)^2) \cdot \poly(1/\epsilon, \log(n d / \delta) ) ).
\end{align*}
\end{lemma}

\begin{proof}
The \textsc{Query} operation has the following two parts:
\begin{itemize}
    \item {\bf Part 1: }Line~\ref{line:q_sketch_gen} takes $O(R L U \cdot \T_{\mathrm{encode}})$ time to call \textsc{Encode} to generate sketches for $q$.
    \item {\bf Part 2: }For every $i \in [n]$:%
    \begin{itemize}
        \item Line~\ref{line:subtract_op} takes $O(R L U \cdot \T_{\mathrm{subtract}})$ time to compute sketch of the difference between %
        $q$ and $x_i$, and store the sketch at index of $n + 1$.
        \item Line~\ref{line:decode} takes $O(R L U \cdot \T_{\mathrm{decode}})$ time to decode the \textsc{BatchHeavyHitter} and get estimated heavy hitters of $q - x_i$.
        \item Line~\ref{line:seletct_start} to Line~\ref{line:seletct_end} takes $O(R L U\cdot 2/\beta)$ time to analyze the \textsc{BatchHeavyHitter} %
        and get the set of indices, where $2/\beta$ is the size of the set. 
        \item Line~\ref{line:compute_size_of_layer_counts} takes $O(L P \cdot 2/\beta)$ time to compute size of the layer sets cut by $\alpha$.
        \item Line~\ref{line:layer_counts_compute_start} to Line~\ref{line:layer_counts_compute_end} takes $O(P L)$ time to compute the estimation of each layer.
    \end{itemize}
    The total running time of this part is:
    \begin{align*}
            & ~ n \cdot (O(R L U \cdot \T_{\mathrm{subtract}}) + O(R L U \cdot \T_{\mathrm{decode}}) + O(R L U\cdot 2/\beta) + O(L P \cdot 2/\beta) + O(LP))\\
        =   & ~ O(n L (R U (\T_{\mathrm{subtract}} + \T_{\mathrm{decode}} + \beta^{-1}) + P \beta^{-1}))
    \end{align*}
    time in total.
\end{itemize}
Taking these two parts together we have the total running time of the \textsc{Query} procedure:
\begin{align*}
        & ~ O(R L U \cdot \T_{\mathrm{encode}}) + O(n L (R U (\T_{\mathrm{subtract}} + \T_{\mathrm{decode}} + \beta^{-1}) + P \beta^{-1})) \\
    =   & ~ O(R L U(\T_{\mathrm{encode}} + n \cdot \T_{\mathrm{subtract}} + n \cdot \T_{\mathrm{decode}} + n \beta^{-1}) + n L P \beta^{-1})\\
    =   & ~ O(\epsilon^{-4} \log^6(d/\delta) \log(n d) (d\log^2(d) + n \beta^{-1}\log^2(d)) + n  \epsilon^{-1} \log^2(d) \beta^{-1})\\
    =   & ~ O(  (d + n \cdot \mmc(l)^2) \cdot \poly(1/\epsilon, \log(n d / \delta) ) )
\end{align*}
where the first step follows from the property of big $O$ notation, the second step follows from the definition of $R,L,U,\T_{\mathrm{encode}}$, $ \T_{\mathrm{encode}}, \T_{\mathrm{subtract}}, \T_{\mathrm{decode}}$ (Theorem~\ref{thm:batch_heavy_hitter})
$,P$, the third step follows from merging the terms.

Thus, we complete the proof.

\end{proof}
Next, we analyze the space usage in our algorithm. We delay the proofs into Section~\ref{sec:space}.
\begin{lemma}[Space complexity of our data structure, informal version of Lemma~\ref{lem:storage_formal}]
\label{lem:storage}
Our data structure (Algorithm~\ref{alg:main_members_init_informal} and  \ref{alg:main_query_informal}) uses space
\begin{align*} 
O( n (d + \mmc(l)^2) \cdot \poly( 1 / \epsilon,  \log(d/\delta) )) .
\end{align*}
\end{lemma}

\section{Correctness of Our Algorithm}
\label{sec:correctness}

The correctness of our distance oracle is proved in the following lemma:

\begin{lemma}[\textsc{Query} correctness]\label{lem:query_correct}

Given a query point $q \in \R^d$, the procedure \textsc{Query} (Algorithm~\ref{alg:main_query}) takes $q$ as input and approximately estimates  $\{\mathrm{dst}_i\}_{i\in [n]}$ the distance between $q$ and every $x_i$ with the norm $l$, such that for every $\mathrm{dst}_i$, with probability at least $1-\delta$, we have
\begin{align*}
    (1-\epsilon) \cdot \|q - x_i\|_{\sym} \le \mathrm{dst}_i \le (1+\epsilon) \cdot \|q - x_i\|_{\sym}
\end{align*}
\end{lemma}

\begin{proof}
    Without loss of generality, we can consider a fixed $i \in [n]$. For simplicity, we denote $x_i$ by $x$.

    Let $v := q - x$. By Lemma~\ref{lem:approx_with_layer}, it is approximated by its layer vector $\L(v)$, namely,
    \begin{align}
    \label{eq:layer_approx}
        \|v\|_{\sym} \le \|\L(v)\|_{\sym} \le (1 + O(\epsilon)) \|v\|_{\sym},
    \end{align}
    where $\|\cdot\|_{\sym}$ is a symmetric norm, denoted also by $l(\cdot)$. %
    
    We assume without loss of generality that $\epsilon \ge 1/\poly(d)$. Our algorithm maintains a data structure that eventually produces a vector $\J(v)$, which is created with the layer sizes $c_1, c_2, \dots, c_P$, where $c$'s denotes the estimated layer sizes output by out data structure, and the $b$'s are the ground truth layer sizes. %
    We will show that with high probability, $\|\J(v)\|_{\sym}$ approximates $\|\L(v)\|_{\sym}$. Specifically, to achieve $(1 \pm \epsilon)$-approximation to $\|v\|_{\sym}$, we set the approximation guarantee of the layer sets (Definition~\ref{def:important_layers}) %
    to be $\epsilon_1:=O(\frac{\epsilon^2}{\log(d)})$ and the importance guarantee to be %
    $\beta_0:=O(\frac{\epsilon^5}{\mmc(l)^2 \log^5(d)})$, where $\mmc(l)$ is defined as Definition~\ref{def:mmc}.
    
    Observe that the number of non-empty layer sets $P = O(\log_\alpha(d)) = O(\log(d)/\epsilon)$. Let $E_\mathrm{succeed}$ denote the event $(1-\epsilon_1)b_k \le c_k \le b_k$. By Lemma~\ref{lem:layer_count_approx}, it happens with high probability $1 - \delta$. Conditioned on this event. %
    
    Denote by $\J(v)$ the vector generated with the layer sizes given by Line~\ref{line:layer_approx}, and by $\L^*(v)$ the vector $\L(v)$ after removing all buckets (Definition~\ref{def:layer_vector_and_bucket}) that are not $\beta$-contributing (Definition~\ref{def:contribut_layer}), %
    and define $\J^*(v)$ similar to $\L^*(v)$, where we set 
        $\beta := \epsilon/P = O(\epsilon^2 / \log(d))$.
    Every $\beta$-contributing layer is necessarily $\beta_0$-important (Definition~\ref{def:important_layers}) by Lemma~\ref{lem:left_beta_contributing} %
    and Lemma~\ref{lem:right_beta_contributing} and therefore satisfies $c_k \ge (1-\epsilon_1)b_k$ (Lemma~\ref{lem:layer_count_approx}). We bound the error of $\|\L^*(v)\|_{\sym}$ %
    by Lemma~\ref{lem:conc_of_contribt}, namely,
    \begin{align*}
             (1-O(\epsilon)) \|\L(v)\|_{\sym}
        \le  (1-O(\log_\alpha d) \cdot \beta) \|\L(v)\|_{\sym} \notag 
        \le  \|\L^*(v)\|_{\sym} \notag 
        \le  \|\L(v)\|_{\sym}.
    \end{align*}
    where the first step follows from the definition of $\beta$, the second step and the third step follow from Lemma~\ref{lem:conc_of_contribt}.
    
    Then, we have
    \begin{align}
    \label{eq:estimation}
             \|\J(v)\|_{\sym}
        \ge & ~ \|\J^*(v)\|_{\sym} \notag \\
        =   & ~ \|\L^*(v) \backslash \L_{k_{1}}(v) \cup \J_{k_{1}}(v) \dots \backslash \L_{k_{\kappa}}(v) \cup \J_{k_{\kappa}}(v)\|_{\sym} \notag \\
        \ge & ~ (1-\epsilon_1)^P  \|\L^*(v)\|_{\sym} \notag \\
        \ge & ~ (1-O(\epsilon)) \|\L^*(v)\|_{\sym}.
    \end{align}
    where the first step follows from  monotonicity (Lemma~\ref{lem:monoton_sym}), the second step follows from the definition of $\J^*(v)$, %
    the third step follows from Lemma~\ref{lem:buckt_approx}, and the fourth step follows from the definition of $\epsilon_1$ and $P$.
    
    Combining Eq.~\eqref{eq:layer_approx} and \eqref{eq:estimation}, %
    we have
    \begin{align}
    \label{eq:conc_of_estimate}
        (1-O(\epsilon)) \cdot \|v\|_{\sym} \le \|\J^*(v)\|_{\sym} \le \|v\|_{\sym},
    \end{align}
    which bounds the error of $\|\J^*(v)\|_{\sym}$ as required. Note that, with Lemma~\ref{lem:conc_of_contribt} we have
    \begin{align}
    \label{eq:conc_of_contribt}
        (1 - O(\epsilon)) \cdot \|\J(v)\|_{\sym} \le \|\J^*(v)\|_{\sym} \le \|\J(v)\|_{\sym}
    \end{align}
    Combining the Eq.\eqref{eq:conc_of_estimate} and \eqref{eq:conc_of_contribt} we have
    \begin{align*}
        (1 - O(\epsilon)) \cdot \|v\|_{\sym} \le \|\J(v)\|_{\sym} \le (1 + O(\epsilon)) \cdot \|v\|_{\sym}
    \end{align*}
    Note that $E_\mathrm{succeed}$ has a failure probability of $\delta$. Thus, we complete the proof. %
\end{proof}
\section{Conclusion}
\label{sec:conclusion}
Similarity search is the backbone of many large-scale applications in machine-learning, optimization, databases and computational geometry. 
Our work strengthens and unifies a long line of work on metric embeddings and sketching, by presenting the first Distance Oracle \emph{for any symmetric norm},  with nearly-optimal query and update times.
The generality of our data structure allows to apply it as a \emph{black-box} for data-driven \emph{learned} symmetric distance metrics \cite{dkj+07} and in various optimization problems involving symmetric distances.

Our work raises several open questions for future study:
\begin{itemize}
    \item The efficiency of our data structure depends on $\mmc(l)$, the concentration property of the symmetric norm. Is this dependence necessary?
    \item Can we generalize our data structure to certain \emph{asymmetric norms} ? %
\end{itemize}

We believe our work is also likely to influence other fundamental problems in high-dimensional optimization and search, e.g, kernel linear regression, geometric sampling and near-neighbor search.

\ifdefined\isarxivversion
\bibliographystyle{alpha}
\bibliography{ref}
\else
\bibliography{ref}

\newcommand{\etalchar}[1]{$^{#1}$}
\begin{thebibliography}{BYJKS04}

\bibitem[AC09]{ac09}
Nir Ailon and Bernard Chazelle.
\newblock The fast johnson--lindenstrauss transform and approximate nearest
  neighbors.
\newblock {\em SIAM Journal on computing}, 39(1):302--322, 2009.

\bibitem[AFS12]{afs12}
Andreas Argyriou, Rina Foygel, and Nathan Srebro.
\newblock Sparse prediction with the $ k $-support norm.
\newblock {\em Advances in Neural Information Processing Systems (NeurIPS)},
  25, 2012.

\bibitem[Alm19]{a19}
Josh Alman.
\newblock An illuminating algorithm for the light bulb problem.
\newblock In {\em 2nd Symposium on Simplicity in Algorithms (SOSA)}. Schloss
  Dagstuhl-Leibniz-Zentrum fuer Informatik, 2019.

\bibitem[ALS{\etalchar{+}}18]{als+18}
Alexandr Andoni, Chengyu Lin, Ying Sheng, Peilin Zhong, and Ruiqi Zhong.
\newblock Subspace embedding and linear regression with orlicz norm.
\newblock In {\em International Conference on Machine Learning (ICML)}, pages
  224--233. PMLR, 2018.

\bibitem[AMS99]{ams99}
Noga Alon, Yossi Matias, and Mario Szegedy.
\newblock The space complexity of approximating the frequency moments.
\newblock {\em Journal of Computer and system sciences}, 58(1):137--147, 1999.

\bibitem[ANN{\etalchar{+}}17]{ann+17}
Alexandr Andoni, Huy~L Nguyen, Aleksandar Nikolov, Ilya Razenshteyn, and Erik
  Waingarten.
\newblock Approximate near neighbors for general symmetric norms.
\newblock In {\em Proceedings of the 49th Annual ACM SIGACT Symposium on Theory
  of Computing (STOC)}, pages 902--913, 2017.

\bibitem[ANN{\etalchar{+}}18]{ann+18}
Alexandr Andoni, Assaf Naor, Aleksandar Nikolov, Ilya Razenshteyn, and Erik
  Waingarten.
\newblock H{\"o}lder homeomorphisms and approximate nearest neighbors.
\newblock In {\em 2018 IEEE 59th Annual Symposium on Foundations of Computer
  Science (FOCS)}, pages 159--169. IEEE, 2018.

\bibitem[AWY14]{awy14}
Amir Abboud, Ryan Williams, and Huacheng Yu.
\newblock More applications of the polynomial method to algorithm design.
\newblock In {\em Proceedings of the twenty-sixth annual ACM-SIAM symposium on
  Discrete algorithms}, pages 218--230. SIAM, 2014.

\bibitem[AWY18]{awy18}
Josh Alman, Joshua~R Wang, and Huacheng Yu.
\newblock Cell-probe lower bounds from online communication complexity.
\newblock In {\em Proceedings of the 50th Annual ACM SIGACT Symposium on Theory
  of Computing}, pages 1003--1012, 2018.

\bibitem[BBC{\etalchar{+}}17]{bbc+17}
Jaros{\l}aw B{\l}asiok, Vladimir Braverman, Stephen~R Chestnut, Robert
  Krauthgamer, and Lin~F Yang.
\newblock Streaming symmetric norms via measure concentration.
\newblock In {\em Proceedings of the 49th Annual ACM SIGACT Symposium on Theory
  of Computing}, pages 716--729, 2017.

\bibitem[Bha97]{b97}
Rajendra Bhatia.
\newblock Matrix analysis, volume 169 of.
\newblock {\em Graduate texts in mathematics}, 1997.

\bibitem[BYJKS04]{bjks04}
Ziv Bar-Yossef, Thathachar~S Jayram, Ravi Kumar, and D~Sivakumar.
\newblock An information statistics approach to data stream and communication
  complexity.
\newblock {\em Journal of Computer and System Sciences}, 68(4):702--732, 2004.

\bibitem[CCF04]{ccf04}
Moses Charikar, Kevin~C. Chen, and Martin Farach{-}Colton.
\newblock Finding frequent items in data streams.
\newblock {\em Theor. Comput. Sci.}, 312(1):3--15, 2004.

\bibitem[Che52]{che52}
Herman Chernoff.
\newblock A measure of asymptotic efficiency for tests of a hypothesis based on
  the sum of observations.
\newblock {\em The Annals of Mathematical Statistics}, pages 493--507, 1952.

\bibitem[CLP{\etalchar{+}}20]{clp20}
Beidi Chen, Zichang Liu, Binghui Peng, Zhaozhuo Xu, Jonathan~Lingjie Li, Tri
  Dao, Zhao Song, Anshumali Shrivastava, and Christopher Re.
\newblock Mongoose: A learnable lsh framework for efficient neural network
  training.
\newblock In {\em International Conference on Learning Representations (ICLR)},
  2020.

\bibitem[CMF{\etalchar{+}}20]{cmf+20}
Beidi Chen, Tharun Medini, James Farwell, Charlie Tai, Anshumali Shrivastava,
  et~al.
\newblock Slide: In defense of smart algorithms over hardware acceleration for
  large-scale deep learning systems.
\newblock {\em Proceedings of Machine Learning and Systems}, 2:291--306, 2020.

\bibitem[CN20]{cn20}
Yeshwanth Cherapanamjeri and Jelani Nelson.
\newblock On adaptive distance estimation.
\newblock In {\em Advances in Neural Information Processing Systems}, 2020.

\bibitem[CS09]{cho09}
Youngmin Cho and Lawrence Saul.
\newblock Kernel methods for deep learning.
\newblock {\em Advances in neural information processing systems}, 22, 2009.

\bibitem[DKJ{\etalchar{+}}07]{dkj+07}
Jason~V Davis, Brian Kulis, Prateek Jain, Suvrit Sra, and Inderjit~S Dhillon.
\newblock Information-theoretic metric learning.
\newblock In {\em Proceedings of the 24th international conference on Machine
  learning}, pages 209--216, 2007.

\bibitem[GMS86]{ms86}
M~Gromov, V~Milman, and G~Schechtman.
\newblock Asymptotic theory of finite dimensional normed spaces, volume 1200 of
  lectures notes in mathematics, 1986.

\bibitem[GPPR04]{gpp04}
Cyril Gavoille, David Peleg, St{\'e}phane P{\'e}rennes, and Ran Raz.
\newblock Distance labeling in graphs.
\newblock {\em Journal of Algorithms}, 53(1):85--112, 2004.

\bibitem[HIKP12]{hikp12a}
Haitham Hassanieh, Piotr Indyk, Dina Katabi, and Eric Price.
\newblock Nearly optimal sparse fourier transform.
\newblock In {\em Proceedings of the forty-fourth annual ACM symposium on
  Theory of computing (STOC)}, pages 563--578, 2012.

\bibitem[HKNS15]{hkns15}
Monika Henzinger, Sebastian Krinninger, Danupon Nanongkai, and Thatchaphol
  Saranurak.
\newblock Unifying and strengthening hardness for dynamic problems via the
  online matrix-vector multiplication conjecture.
\newblock In {\em Proceedings of the forty-seventh annual ACM symposium on
  Theory of computing}, pages 21--30, 2015.

\bibitem[Ind06]{i06}
Piotr Indyk.
\newblock Stable distributions, pseudorandom generators, embeddings, and data
  stream computation.
\newblock {\em Journal of the ACM (JACM)}, 53(3):307--323, 2006.

\bibitem[IRW17]{irw17}
Piotr Indyk, Ilya Razenshteyn, and Tal Wagner.
\newblock Practical data-dependent metric compression with provable guarantees.
\newblock {\em Advances in Neural Information Processing Systems}, 30, 2017.

\bibitem[IW05]{iw05}
Piotr Indyk and David Woodruff.
\newblock Optimal approximations of the frequency moments of data streams.
\newblock In {\em Proceedings of the thirty-seventh annual ACM symposium on
  Theory of computing}, pages 202--208, 2005.

\bibitem[JKDG08]{jkdg08}
Prateek Jain, Brian Kulis, Inderjit~S Dhillon, and Kristen Grauman.
\newblock Online metric learning and fast similarity search.
\newblock In {\em NIPS}, volume~8, pages 761--768. Citeseer, 2008.

\bibitem[JL84]{jl84}
William~B Johnson and Joram Lindenstrauss.
\newblock Extensions of lipschitz mappings into a hilbert space.
\newblock {\em Contemporary mathematics}, 26(189-206):1, 1984.

\bibitem[KNPW11]{knpw11}
Daniel~M Kane, Jelani Nelson, Ely Porat, and David~P Woodruff.
\newblock Fast moment estimation in data streams in optimal space.
\newblock In {\em Proceedings of the forty-third annual ACM symposium on Theory
  of computing}, pages 745--754, 2011.

\bibitem[KNW10]{knw10}
Daniel~M Kane, Jelani Nelson, and David~P Woodruff.
\newblock On the exact space complexity of sketching and streaming small norms.
\newblock In {\em Proceedings of the twenty-first annual ACM-SIAM symposium on
  Discrete Algorithms}, pages 1161--1178. SIAM, 2010.

\bibitem[LDFU13]{ldfu13}
Yichao Lu, Paramveer Dhillon, Dean~P Foster, and Lyle Ungar.
\newblock Faster ridge regression via the subsampled randomized hadamard
  transform.
\newblock In {\em Advances in neural information processing systems (NIPS)},
  pages 369--377, 2013.

\bibitem[LNNT16]{lnnt16}
Kasper~Green Larsen, Jelani Nelson, Huy~L Nguy{\^e}n, and Mikkel Thorup.
\newblock Heavy hitters via cluster-preserving clustering.
\newblock In {\em 2016 IEEE 57th Annual Symposium on Foundations of Computer
  Science (FOCS)}, pages 61--70. IEEE, 2016.

\bibitem[LNRW19]{lnrw19}
Jerry Li, Aleksandar Nikolov, Ilya Razenshteyn, and Erik Waingarten.
\newblock On mean estimation for general norms with statistical queries.
\newblock In {\em Conference on Learning Theory (COLT)}, pages 2158--2172.
  PMLR, 2019.

\bibitem[LSV18]{lsv18}
Yin~Tat Lee, Zhao Song, and Santosh~S Vempala.
\newblock Algorithmic theory of odes and sampling from well-conditioned
  logconcave densities.
\newblock {\em arXiv preprint arXiv:1812.06243}, 2018.

\bibitem[LW17]{lw17}
Kasper~Green Larsen and Ryan Williams.
\newblock Faster online matrix-vector multiplication.
\newblock In {\em Proceedings of the Twenty-Eighth Annual ACM-SIAM Symposium on
  Discrete Algorithms}, pages 2182--2189. SIAM, 2017.

\bibitem[MMR22]{mmr19}
Konstantin Makarychev, Yury Makarychev, and Ilya Razenshteyn.
\newblock Performance of johnson--lindenstrauss transform for k-means and
  k-medians clustering.
\newblock {\em SIAM Journal on Computing}, (0):STOC19--269, 2022.

\bibitem[MPS14]{mps14}
Andrew~M McDonald, Massimiliano Pontil, and Dimitris Stamos.
\newblock Spectral k-support norm regularization.
\newblock {\em Advances in neural information processing systems}, 27, 2014.

\bibitem[NS19]{ns19}
Vasileios Nakos and Zhao Song.
\newblock Stronger l2/l2 compressed sensing; without iterating.
\newblock In {\em Proceedings of the 51st Annual ACM SIGACT Symposium on Theory
  of Computing}, pages 289--297, 2019.

\bibitem[Pag13]{p13}
Rasmus Pagh.
\newblock Compressed matrix multiplication.
\newblock {\em ACM Transactions on Computation Theory (TOCT)}, 5(3):1--17,
  2013.

\bibitem[Pel00]{pel00}
David Peleg.
\newblock Proximity-preserving labeling schemes.
\newblock {\em Journal of Graph Theory}, 33(3):167--176, 2000.

\bibitem[Pri11]{pri11}
Eric Price.
\newblock Efficient sketches for the set query problem.
\newblock In {\em Proceedings of the Twenty-Second Annual {ACM-SIAM} Symposium
  on Discrete Algorithms, {SODA} 2011, San Francisco, California, USA, January
  23-25, 2011}, pages 41--56, 2011.

\bibitem[SS02]{ss02}
Michael Saks and Xiaodong Sun.
\newblock Space lower bounds for distance approximation in the data stream
  model.
\newblock In {\em Proceedings of the thiry-fourth annual ACM symposium on
  Theory of computing}, pages 360--369, 2002.

\bibitem[SSSSC11]{ssy11}
Shai Shalev-Shwartz, Yoram Singer, Nathan Srebro, and Andrew Cotter.
\newblock Pegasos: Primal estimated sub-gradient solver for svm.
\newblock {\em Mathematical programming}, 127(1):3--30, 2011.

\bibitem[SSWZ22]{sswz22}
Zhao Song, Baocheng Sun, Omri Weinstein, and Ruizhe Zhang.
\newblock Sparse fourier transform over lattices: A unified approach to signal
  reconstruction.
\newblock http://arxiv.org/abs/2205.00658, 2022.

\bibitem[SSX21]{ssx21}
Anshumali Shrivastava, Zhao Song, and Zhaozhuo Xu.
\newblock Sublinear least-squares value iteration via locality sensitive
  hashing.
\newblock {\em arXiv preprint arXiv:2105.08285}, 2021.

\bibitem[SWY{\etalchar{+}}19]{swy+19}
Zhao Song, Ruosong Wang, Lin Yang, Hongyang Zhang, and Peilin Zhong.
\newblock Efficient symmetric norm regression via linear sketching.
\newblock {\em Advances in Neural Information Processing Systems}, 32, 2019.

\bibitem[SWZ19]{swz19}
Zhao Song, David Woodruff, and Peilin Zhong.
\newblock Towards a zero-one law for column subset selection.
\newblock {\em Advances in Neural Information Processing Systems}, 32, 2019.

\bibitem[SXZ22]{sxz22}
Zhao Song, Zhaozhuo Xu, and Lichen Zhang.
\newblock Speeding up sparsification with inner product search data structures.
\newblock {\em arXiv preprint arXiv:2204.03209}, 2022.

\bibitem[TZ12]{tz12}
Mikkel Thorup and Yin Zhang.
\newblock Tabulation-based 5-independent hashing with applications to linear
  probing and second moment estimation.
\newblock {\em {SIAM} J. Comput.}, 41(2):293--331, 2012.

\bibitem[Val12]{v12}
Gregory Valiant.
\newblock Finding correlations in subquadratic time, with applications to
  learning parities and juntas.
\newblock In {\em 2012 IEEE 53rd Annual Symposium on Foundations of Computer
  Science (FOCS)}, pages 11--20. IEEE, 2012.

\bibitem[WP11]{wp11}
Oren Weimann and David Peleg.
\newblock A note on exact distance labeling.
\newblock {\em Information processing letters}, 111(14):671--673, 2011.

\bibitem[XSS21]{xss21}
Zhaozhuo Xu, Zhao Song, and Anshumali Shrivastava.
\newblock Breaking the linear iteration cost barrier for some well-known
  conditional gradient methods using maxip data-structures.
\newblock {\em Advances in Neural Information Processing Systems (NeurIPS)},
  34, 2021.

\end{thebibliography}
\bibliographystyle{alpha}
\input{checklist}
\fi

\newpage
\onecolumn
\appendix

\newpage

\paragraph{Roadmap.}

We divide the appendix into the following sections. Section~\ref{sec:preliminary} gives the preliminaries for our work.  Section~\ref{sec:sym_norms} introduces some useful properties of symmetric norms. Section~\ref{sec:layer_approx} gives the proofs for the layer approximation technique. Section~\ref{sec:data_structure} states the formal version of our main theorem and algorithms. Section~\ref{sec:more_detail_time} gives more details about the time complexity proofs.  Section~\ref{sec:detail_correctness} gives some details about the correctness proofs.  Section~\ref{sec:space} shows the space complexity of our algorithm. Section~\ref{sec:lower_bound} states a streaming lower bound for the norm estimation problem and shows that our result is tight in this case. Section~\ref{sec:sparse_recv_tool} gives an instantiation of the sparse recovery data structure (\textsc{BatchHeavyHitter}) that simplifies the analysis in prior work. %

\section{Preliminaries}
\label{sec:preliminary}

In Section~\ref{sec:notation}, we define the notations we use in this paper. In Section~\ref{sec:prob_tools}, we introduce some probability tools. In Section~\ref{sec:stable}, we define $p$-stable distributions.

\subsection{Notations}
\label{sec:notation}

For any positive integer $n$, we use $[n]$ to denote a set $\{1,2,\cdots,n\}$. We use $\E[]$ to denote the expectation. We use $\Pr[]$ to denote the probability. We use $\var[]$ to denote the variance. We define the unit vector $\xi ^{(d_1)}:=\frac{1}{\sqrt{d_1}}(1,1,1,,\dots,1,0,\dots,0)\in\mathbb{R}^d$, for any $d_1\in[d]$, which has $d_1$ nonzero coordinates. We abuse the notation to write $\xi^{(d_1)}\in \mathbb{R}^{d_1}$ by removing zero coordinates, and vice-versa by appending zeros. 
We use $\| x \|_2$ to denote entry-wise $\ell_2$ norm of a vector. We use $\| x\|_{\sym}$ to denote the symmetric norm of a vector $x$.

We define tail as follows
\begin{definition}[Tail of a vector]\label{def:tail}
For a given $x \in \R^d$ and an integer $k$, we use $x_{\ov{[k]}}$ or $x_{\tail(k)}$ to denote the vector that without largest top-$k$ values (in absolute).
\end{definition}

\subsection{Probability Tools}
\label{sec:prob_tools}

We state some useful inequalities in probability theory in below.
\begin{lemma}[Chernoff bound \cite{che52}]
\label{lem:Chernoff_bound}
Let $X = \sum_{i =1}^n X_i$, where $X_i = 1$ with probability $p_i$ and $X_i = 0$ with probability $1 - p_i$, and all $X_i$ are independent. Let $\mu = \E[X] = \sum_{i = 1}^n p_i$. Then
\begin{itemize}
    \item $\Pr[X \ge (1+\delta)\mu] \le \exp(-\delta^2\mu/3)$, $\forall \delta > 0$;
    \item $\Pr[X \le (1-\delta)\mu] \le \exp(-\delta^2\mu/2)$, $\forall 0 < \delta < 1$.
\end{itemize}
\end{lemma}

\begin{lemma}[Chebyshev's inequality]\label{lem:chebyshev_inequality}
Let $X$ be a random variable with finite expected value $\mu$ and finite non-zero variance $\sigma^2$. Then for any real number $k > 0$, 
\begin{align*}
    \Pr[ | X - \mu | \geq k \sigma ] \leq \frac{1}{k^2}.
\end{align*}
\end{lemma}

\begin{theorem}[Levy's isoperimetric inequality, \cite{ms86}]
\label{thm:Levy's_inequality}
For a continuous function $f:S^{d-1}\rightarrow \mathbb{R}$, Let $M_f$ be the median of $f$, i.e., $\mu(\{x:f(x)\le M_f\})\ge 1/2$ and $\mu(\{x:f(x)\ge M_f\})\ge 1/2$, where $\mu(\cdot)$ is the Haar probability measure on the unit sphere $S^{d-1}$. Then
\begin{align*}
    \mu(\{x:f(x)=M_f\}_\epsilon)\ge 1-\sqrt{\pi/2} e^{-\epsilon^2 d/2},
\end{align*}
where for a set $A\subset S^{d-1}$ we denote $A_\epsilon:=\{x:l_2(x,A)\le \epsilon\}$ and $l_2(x,A):=\inf_{y\in A}\|x-y\|_2$.
\end{theorem}

\subsection{Stable Distributions}\label{sec:stable}
We define $p$-stable distributions.
\begin{definition}[\cite{i06}]
\label{def:p_stable_distribution}
    A distribution $\mathcal{D}_p$ is called p-stable, if there exists $p \ge 0$ such that for any $d$ real numbers $a_1, a_2, \dots, a_d$ and i.i.d. variables $x_1, x_2, \dots, x_d$ from distribution $\mathcal{D}_p$. the random variable $\sum_{i = 1}^d a_ix_i$ has the same distribution as the variable $\|a\|_p y$, where y is a random variable from distribution $\mathcal{D}_p$.
\end{definition}

\section{Symmetric Norms}
\label{sec:sym_norms}

In this section, we give several technical tools for the symmetric norms.  In Section~\ref{sec:property_of_sym_norm}, we show the monotonicity of symmetric norms. In Section~\ref{sec:conc_of_measure_result}, we show the concentration properties of symmetric norms. In Section~\ref{sec:flat_lemma}, we show some properties of the median of symmetric norm. %

\subsection{Monotonicity Property of Symmetric Norm}
\label{sec:property_of_sym_norm}

\begin{lemma}[Monotonicity of Symmetric Norms,  see e.g. Proposition IV.1.1 in \cite{b97}
]
\label{lem:monoton_sym}
If $\| \cdot \|_{\sym}$ is a symmetric norm and $x,y\in\mathbb{R}^d$  
satisfy that for all $i \in [d]$ 
, $|x_i|\le|y_i|$, then $\| x \|_{\sym} \le \| y \|_{\sym}$.

\end{lemma}

\subsection{Concentration Property of Symmetric Norms}
\label{sec:conc_of_measure_result}

Let us give the results of concentration of measure as followed. The following tools and proofs can be found in \cite{bbc+17}. However, for the completeness, we state them in below.

\begin{lemma}[Concentration of $M_l$%
]
\label{lem:concentration_M_l}
    For every norm $l$ on $\mathbb{R}^d$, if $x\in S^{d-1}$ is drawn uniformly at random according to Haar measure on the sphere, then 
    \begin{align*}
        \Pr[\|x\|_{l}-M_l|>\frac{2\mathfrak{b}_l}{\sqrt{d}}]<\frac{1}{3}
    \end{align*}
\end{lemma}

\begin{proof}
With Theorem~\ref{thm:Levy's_inequality}, we know that, for a random $x$ distributed according to the Haar measure on the $l_2$-sphere, with probability at least $1-\sqrt{\pi/2}e^{-2}>\frac{2}{3}$, we can always find some $y\in S^{d-1}$, such that
\begin{align*}
    \|x-y\|_2 & \le \frac{2}{\sqrt{d}},\text{ and} \\
    \|y\|_{l} & = M_l.
\end{align*}
If we view the norm $l$ as a function, then it is obvious that it is $\mathfrak{b}_l$-Lipschitz with respect to $\|\cdot\|_2$, so that we have
\begin{align*}
    |\|x\|_{l}-M_l| = & ~ |\|x\|_{l}-\|y\|_{l}| \\
    \le & ~ \|x - y\|_{l} \\
    \le & ~ \mathfrak{b}_l \|x-y\| \\
    \le & ~ \frac{2\mathfrak{b}_l}{\sqrt{d}}
\end{align*}
where the first step follows from the definition of $y$, the second step follows from triangle inequality, the third step follows from that norm $l$ is $\mathfrak{b}_l$-Lipschitz with respect to $\|\cdot\|_2$, and the last step follows from the definition of $y$. 

Thus, we complete the proof.
\end{proof}

\begin{lemma}[Concentration inequalities for norms
]
\label{lem:conc_of_norm}
For every $d > 0$, and norm $l$ on $\R^n$, there is a vector $x\in S^{d-1}$ satisfying
\begin{itemize}
    \item $|\|x\|_{\infty}-M_{l_\infty^{(d)}}|\le 2/\sqrt{d}$
    \item $|\|x\|_{l}-M_{l^{(d)}}|\le 2\mathfrak{b}_{l^{(d)}}/\sqrt{d}$, and
    \item $|\{i:|x_i|>\frac{1}{K\sqrt{d}}\}|>\frac{d}{2}$ for some universal constant $K$.
\end{itemize}
\end{lemma}

\begin{proof}
Let $x$ be drawn uniformly randomly from a unit sphere. From  Lemma~\ref{lem:concentration_M_l}, we have
\begin{align*}
    \Pr[|\|x\|_{l}-M_{l_\infty^{(d)}}|>\frac{2}{\sqrt{d}}] & ~ <\frac{1}{3}\\
    and\quad\Pr[|\|x\|_{l}-M_{l^{(d)}}|>\frac{2\mathfrak{b}_{l}}{\sqrt{d}}] & ~ < \frac{1}{3}.
\end{align*}
Define $\tau(x,t):=|\{i \in [d] : |x_i|<t\}|$. Now we are giving the proof to show that, for a constant $K$, over a choice of $x$, we have $\tau(x,\frac{1}{K\sqrt{d}})<\frac{d}{2}$ with probability lager than $2/3$ .

Let's consider an random vector $z\in\mathbb{R}^d$, such that each entry $z_i$ is independent standard normal random variable. It is well known that $\frac{z}{\|z\|_2}$ is distributed uniformly over a sphere, so it has the same distribution as $x$. There is a universal constant $K_1$ such that
\begin{align*}
    \Pr[\|z\|_2 > K_1 \sqrt{d}] < \frac{1}{6},
\end{align*}
and similarly, there is a constant $K_2$, such that
\begin{align*}
    \Pr[|z_i|<\frac{1}{K_2}]<\frac{1}{12}.
\end{align*}
Therefore, by Markov bound we have
\begin{align*}
    \Pr[\tau(z,\frac{1}{K_2}) > \frac{d}{2}] < \frac{1}{6}.
\end{align*}
By union bound, with probability larger than $2/3$, it holds simultaneously that
\begin{align*}
    \|z\|_2  ~ \le K_1 \sqrt{d}\quad \text{and} \quad
    \tau(z, \frac{1}{K_2})  ~ < \frac{d}{2},
\end{align*}
which imply:
\begin{align*}
    \tau(\frac{z}{\|z\|_2},\frac{1}{K_1 K_2 \sqrt{d}}) < \frac{d}{2}.
\end{align*}
Now, by union bound, a random vector $x$ satisfies all of the conditions in the statement of the lemma with positive probability.

Thus, we complete the proof.
\end{proof}

\subsection{Median of Symmetric Norm}
\label{sec:flat_lemma}

We state a well-known fact before introducing Lemma~\ref{lem:flat_median_lem}.
\begin{fact}[Concentration for median on infinity norm, \cite{ms86}] 
\label{fact:flat_median_l_infty}
There are absolute constants $0<\gamma_1<\gamma_2$ such that for every integer $d \ge 1$,
\begin{align*}
    \gamma_1\sqrt{\log (d)/d} \le M_{l^{(d)}_\infty} \le \gamma_2\sqrt{\log (d)/d}
\end{align*}
\end{fact}

We now give the following Lemma, which says that the $l$-norm of the (normalized) unit vector $\xi^{(d)}$ is closely related to the median of the norm. We are considering this vector because that it is related to a single layer of $\L$.

\begin{lemma}[Flat Median Lemma %
]
\label{lem:flat_median_lem}
    Let $l:\mathbb{R}^d\rightarrow \mathbb{R}$ be a symmetric norm. Then
    \begin{align*}
        \lambda_1 M_l /\sqrt{\log (d)} \le l(\xi^{(d)}) \le \lambda_2 M_l
    \end{align*}
    where $\lambda_1, \lambda_2 > 0$ are absolute constants.
\end{lemma}

We notice that the first inequality is tight for $l_\infty$.

\begin{proof}
Using Lemma~\ref{lem:conc_of_norm}, we can find a vector $x\in S^{d-1}$ and a constant $\lambda > 0$ satisfying
\begin{itemize}
    \item $|\|x\|_{\infty}-M_{l_\infty}|\le \lambda\sqrt{1/d}$
    \item $|\|x\|_{l}-M_{l}|\le \lambda \mathfrak{b}_{l}/\sqrt{d}$, and
    \item $|\{i:|x_i|>\frac{1}{K\sqrt{d}}\}|>\frac{d}{2}$ for some universal constant $K$.
\end{itemize}
With Fact~\ref{fact:flat_median_l_infty}, $M_{l_\infty} = \Theta(\sqrt{\log(d)/d})$. On the other hand, let $\mmc(l)$ be defined as Definition~\ref{def:mmc}, we have 
\begin{align*}
    \mmc(l) \le \gamma \sqrt{d}
\end{align*}
for sufficiently small $\gamma$, thus
\begin{align*}
    \frac{\lambda\mathfrak{b}_l}{\sqrt{d}} < M_l.
\end{align*}
So with constants $\gamma_1,\gamma_2>0$, we have
\begin{align*}
    \gamma_1 M_l & ~ \le \|x\|_l \le \gamma_2 M_l\\
    and\quad \gamma_1\sqrt{\log(d)/d} & ~ \le \|x\|_{\infty} \le \gamma_2\sqrt{\log(d)/d}.
\end{align*}
So that we have
\begin{align*}
    |x| \le \gamma_2\sqrt{\log (d)}\cdot\xi^{(d)},
\end{align*}
where the inequality is entry-wise, and with monotonicity of symmetric norms (Lemma~\ref{lem:monoton_sym}), we have
\begin{align*}
    \gamma_1 \le \|x\|_l \le \gamma_2 \sqrt{\log (d)}\cdot \|\xi^{(d)}\|_l.
\end{align*}
Now we move to the second condition of the lemma, we first let $M = \{i:|x_i|>\frac{1}{K\sqrt{d}}\}$. As $|M|>\frac{d}{2}$, we can find a  permutation $\pi$ satisfying
\begin{align*}
    [d]-M \in \pi(M).
\end{align*}
We define a vector $\pi(x)$ to be the vector by applying the permutation $\pi$ to each entry of $x$. Denote by $|x|$ the vector of taking absolute value of $x$ entry-wise.
Notice $|x|+\pi(|x|) > \frac{\xi^{(d)}}{K}$ entry-wise, so that we have
\begin{align*}
    \frac{1}{K}\cdot \|\xi^{(d)}\|_l
    \le & ~ \||x| + \pi(|x|)\|_l \\
    \le & ~ \||x|\|_l + \|\pi(|x|)\|_l \\
    =   & ~ 2\|x\|_l \\
    \le  & ~ 2\gamma_2 M_l,
\end{align*}
where the first step follows from the monotonicity of symmetric norm (Lemma~\ref{lem:monoton_sym}), the second step follows from the triangle inequality, the third step follows from the definition of $\pi$, and the last step follows from the definition of $\gamma_2$.

Thus, we complete the proof.
\end{proof}

The following lemma shows the monotonicity of the median (in $d$), a very useful property in the norm approximation.

\begin{lemma}[Monotonicity of Median%
]
\label{lem:monoton_of_med}
Let $l:\mathbb{R}^d\rightarrow \mathbb{R}$ be a symmetric norm. For all $d_1\le d_2$, where $d_1,d_2\in[n]$, let $\mmc(l)$ be defined as Definition~\ref{def:mmc}, we have
\begin{align*}
    M_{l^{(d_1)}}\le\lambda \cdot \mmc(l) \cdot \sqrt{\log (d_1)} M_{l^{(d_2)}},
\end{align*}
where $\lambda>0$ is an absolute constant.
\end{lemma}
\begin{proof}
By Lemma~\ref{lem:flat_median_lem} and the fact that $\xi^{(d_1)}$ is also a vector in $S^{d_2-1}$,
\begin{align*}
    \lambda M_{l^{(d_1)}}/\sqrt{\log (d_1)}\le \|\xi^{(d_1)}\|_{l} \le \mathfrak{b}_{l^{(d_2)}} \le \mmc(l)M_{l^{(d_2)}}.
\end{align*}
Thus, we complete the proof.
\end{proof}

\section{Analysis of Layer Approximation}
\label{sec:layer_approx}

In this section, we show how to estimate the symmetric norm $l(\cdot)$ of a vector using the layer vectors. Section~\ref{sec:layer_vector_and_important_layer} gives the definitions of layer vectors and important layers. Section~\ref{sec:approx_with_layer} shows that we can approximate the exact value of the norm and the layer vector. Section~\ref{sec:contribut_layer} defines the contributing layer and shows its concentration property. Section~\ref{sec:ineq_of_contribut_layers} proves that contributing layers are also important. 

Throughout this section, let $\epsilon \in (0,1)$ be the precision, $\alpha > 0$ and $\beta \in (0,1]$ be some parameters depending on $d$, $\epsilon$ and $\mmc(l)$, where $\mmc(l)$ is defined as Definition~\ref{def:mmc}. Furthermore, we assume $\mmc(l)\le\gamma\sqrt{d}$, for constant parameter $0 \le \gamma\ll 1/2$ small enough.\footnote{We note that beyond this regime,  the streaming lower bound in Theorem~\ref{thm:low_bound_bbc+17} implies that a linear-sized memory (time) is required to approximate the norm.}

\subsection{Layer Vectors and Important Layers}
\label{sec:layer_vector_and_important_layer}

\begin{definition}[Important Layers]

\label{def:important_layers}
For $v\in \mathbb{R}^d$, define \text{layer} $i \in \mathbb{N}_+ $ as 
\begin{align*}
B_i:=\{j\in[d]:\alpha^{i-1} < |v_j|\le\alpha^i\},
\end{align*}
and denote its size by $b_i:=|B_i|$. We denote the number of non-zero $b_i$'s by $t$, the number of non-empty layers. And we say that layer-$i$ is $\beta$-important if 

\begin{itemize}
    \item $b_i > \beta\cdot \sum_{j=i+1}^t b_j $
    \item $b_i \alpha^{2i} \ge \beta\cdot \sum_{j\in [i]} b_j \alpha^{2j}$
\end{itemize}
\end{definition}

With out loss of generality, We restrict the entries of the vector $v$ to be in $[-m, m]$, and that $m=\poly(d)$. Then we know that the number of non-zero $b_i$'s is at most $P=O(\log_\alpha (d))$. In the view of $\ell_2$-norm, for an arbitrary vector $v \in \R^d$, if we normalize it to a unit vector, then the absolute value of each non-zero entry is at least $1/\poly(d)$. In order to simplify our analysis and algorithm for approximating $\| v \|_{\sym}$, we introduce the notations we use as follows.

\begin{definition}[Layer Vectors and Buckets%
]
\label{def:layer_vector_and_bucket}
For each $i \in [P]$, let $\alpha^i \cdot {\bf 1}_{b_i} \in \R^{b_i}$ denote a vector that has length $b_i$ and every entry is $\alpha^i$. Define the layer vector for $v\in\mathbb{R}^d$ with integer coordinates to be %
\begin{align*}
    \L(v) := (\alpha^1 \cdot {\bf 1}_{b_1}, \alpha^2 \cdot {\bf 1}_{b_2}, \cdots, \alpha^P \cdot {\bf 1}_{b_P}, 0 \cdot {\bf 1}_{d - \sum_{j \in [P]} b_j}) \in \R^d;
\end{align*}
and define the $i$-th bucket of $\L(v)$ to be
\begin{align*}
    \L_i(v) := (0 \cdot {\bf 1}_{b_1+b_2+\dots+b_{i-1}}, \alpha^i \cdot {\bf 1}_{b_i}, 0 \cdot {\bf 1}_{d - \sum_{j \in [i]} b_j}) \in \R^d;
\end{align*}
We also define %
$\J(v)$ and $\J_i(v)$ as above by replacing $\{b_i\}$ with the approximated values $\{c_i\}$. %
Denote $\L (v)\backslash \L_i(v)$ as the vector with the $i$-th bucket of $\L(v)$ replaced by 0. We also denote $(\L (v)\backslash \L_i(v)) \cup \J_i(v)$ as the vector by replacing the $i$-th bucket of $\L(v)$ with $\J_i(v)$, i.e., 
\begin{align*}
     (\L(v) \backslash \L_{i}(v)) \cup \J_{i}(v) := (\alpha^1 \cdot {\bf 1}_{b_1}, \alpha^2 \cdot {\bf 1}_{b_2}, \cdots, \alpha^i \cdot {\bf 1}_{c_i}, \cdots, \alpha^P \cdot {\bf 1}_{b_P}, 0 \cdot {\bf 1}_{d - \sum_{j \in [P]} b_j + b_i - c_i}) \in \R^d;
\end{align*}

\end{definition}

\subsection{Approximated Layers Provides a Good Norm Approximation}
\label{sec:approx_with_layer}

We now proof that,  $\|v\|_{\sym}$ can be approximated by using layer vector $V$. We first choose a base to be $\alpha:=1+O(\epsilon)$.

\begin{lemma}[Approximattion with Layer Vector%
]
\label{lem:approx_with_layer}
For all $v\in\mathbb{R}^d$, we have
\begin{align*}
    \|\L(v)\|_{\sym}/\alpha \le \|v\|_{\sym} \le \|\L(v)\|_{\sym}.
\end{align*}
\end{lemma}

\begin{proof}
The lemma follows from the monotonicity of symmetric norms(Lemma~\ref{lem:monoton_sym}) directly.
\end{proof}

The next key lemma shows that $\|\J(v)\|_{\sym}$ is a good approximation to $\|\L(v)\|_{\sym}$.

\begin{lemma}[Bucket Approximation %
]
\label{lem:buckt_approx}
For every layer $i \in [P]$, %
\begin{itemize} 
\item if $c_i\le b_i$, then $\| (\L(v) \backslash \L_{i}(v)) \cup \J_{i}(v) \|_{\sym} \leq \| \L(v) \|_{\sym}$; 
\item if $c_i\ge (1-\epsilon)b_i$, then $\| (\L(v) \backslash \L_{i}(v)) \cup \J_{i}(v) \|_{\sym} \ge (1-\epsilon) \| \L(v) \|_{\sym}$.
\end{itemize}
\end{lemma}

\begin{proof}
With the monotonicity of norm (Lemma~\ref{lem:monoton_sym}), the upper bound is quite obvious. So we just focus on the lower bound. 
Let us take the vector%
\begin{align*}
    \J_i(v) := (0 \cdot {\bf 1}_{c_1+c_2+\dots+c_{i-1}}, \alpha^i \cdot {\bf 1}_{c_i}, 0, \cdots, 0) \in \R^d;
\end{align*}
Here we define $\mathcal{K}(v):=\L(v) - \L_i(v)$. Then notice that $\mathcal{K}(v) + \J_i(v)$ is a permutation of the vector $(\L(v) \backslash \L_{i}(v)) \cup \J_{i}(v)$. We will then show that, under assumptions of the lemma, we have 
\begin{align*}
    \|\mathcal{K}(v) + \J_i(v)\|_{\sym} \geq (c_{i} / b_{i}) \|\L(v)\|_{\sym}.
\end{align*}
Assume a vector $v\in \mathbb{R}^d$ and a permutation $\pi\in\Sigma_d$, we define a vector $\pi(x)$ to be the vector by applying the permutation $\pi$ to each entry of $x$. Using the property of the symmetric norm, we have that $\|v\|_{\sym} = \|\pi(v)\|_{\sym}$. Consider a set of permutations that are cyclic shifts over the non-zero coordinates of $\L_i$, and do not move any other coordinates. That is, there is exactly $b_i$ permutations in $S$, and for every 
$\pi\in S$, we have $\pi(\mathcal{K}(v))=\mathcal{K}(v)$. By the construction of $S$, we have,
\begin{align*}
    \sum_{\pi\in S}\pi(\J_i(v)) = c_i \L_i(v)
\end{align*}
and therefore $\sum_{\pi\in S}\pi(\mathcal{K}(v)+\J_i(v))=c_i \L_i(v) + b_i \mathcal{K}(v)$. As the vectors $\L_i(v)$ and $\mathcal{K}(v)$ have disjoint support, by monotonicity of symmetric norm (Lemma~\ref{lem:monoton_sym}) with respect to each coordinates we can deduce $\|c_i \L_i(v) + b_i \mathcal{K}(v)\|_{\sym} \ge \|c_i (\L_i(v) + \mathcal{K}(v))\|_{\sym}$. By plugging those together,
\begin{align*}
    c_i \|\L_i(v) + \mathcal{K}(v)\|_{\sym}
    \le & ~\|c_i \L_i(v) + b_i \mathcal{K}(v)\|_{\sym} \\
    = & ~ \|\sum_{\pi\in S} \pi(\J_i(v) + \mathcal{K}(v))\|_{\sym}\\
    \le & ~ \sum_{\pi\in S} \|\pi(\J_i(v) + \mathcal{K}(v))\|_{\sym}\\
    = & ~ b_i \|\pi(\J_i(v) + \mathcal{K}(v))\|_{\sym}
\end{align*}
where the first step follows from $c_i \|\L_i(v) + \mathcal{K}(v)\|_{\sym} =  \|c_i(\L_i(v) + \mathcal{K}(v))\|_{\sym}$ and the monotonicity of the norm $l$, the second step follows from $\sum_{\pi\in S}\pi(\J_i(v)+\mathcal{K}(v))=c_i \L_i(v) + b_i \mathcal{K}(v)$, the third step follows from triangle inequality, and the last step follows from the property of symmetric norm and $|S| = b_i$. 

Hence, 
\begin{align*}
\|\J_i(v) + \mathcal{K}(v)\|_{\sym} \ge \frac{c_i}{b_i}\|\L(v)\|_{\sym}  \ge (1-\epsilon)\|\L(v)\|_{\sym},
\end{align*}
Thus, we complete the proof.
\end{proof}

\subsection{Contributing Layers}
\label{sec:contribut_layer}

\begin{definition}[Contributing Layers%
]
\label{def:contribut_layer}
For $i \in [P]$,
layer $i$ is called $\beta$-contributing if 
\begin{align*} 
\|\L_i(v)\|_{\sym} \ge \beta \|\L(v)\|_{\sym}.
\end{align*}
\end{definition}

\begin{lemma}[Concentration with contributing layers%
]
\label{lem:conc_of_contribt}
Let $\L^*(v)$ be the vector obtained from $V$ by removing all layers that are not $\beta$-contributing. Then
\begin{align*}
(1-O(\log_\alpha (d))\cdot \beta) \cdot \|\L(v)\|_{\sym} \le \|\L^*(v)\|_{\sym}\le \|\L(v)\|_{\sym}.
\end{align*}
\end{lemma}

\begin{proof}
Let $i_1,\dots,i_k\in[P]$ be the layers that are not $\beta$- contributing. 

Then we apply the triangle inequality and have,
\begin{align*}
    \|\L(v)\|_{\sym} 
    \ge & ~ \|\L(v)\|_{\sym} - \|\L_{i_1}(v)\|_{\sym} - \dots - \|\L_{i_k}(v)\|_{\sym} \\
    \ge & ~ (1-k_\beta)\|\L(v)\|_{\sym}
\end{align*}
The proof follows by bounding $k$ by $P = O(\log_\alpha (n))$, which is the total number of non-zero $b_i$'s. 
\end{proof}

\subsection{Contributing Layers Are Important}
\label{sec:ineq_of_contribut_layers}
In this section, we give two lemmas to show that every $\beta$-contributing layer (Definition~\ref{def:contribut_layer}) is $\beta'$-important (Definition~\ref{def:important_layers}), where $\beta'$ is depending on $\mmc(l)$ (Definition~\ref{def:mmc}). The first property of the important layer is proved in Lemma~\ref{lem:left_beta_contributing}, and the second property is proved in Lemma~\ref{lem:right_beta_contributing}.

\begin{lemma}[Importance of contributing layers (Part 1)%
]
\label{lem:left_beta_contributing}
    For $i \in [P]$, if layer $i$ is $\beta$-contributing, then for some absolute constant $\lambda>0$, we have
    \begin{align*}
        b_i\ge\frac{\lambda \beta^2}{\mmc(l)^2\log^2 (d)}\cdot\sum_{j=i+1}^P b_j,
    \end{align*}
    where $\mmc(l)$ is defined as Definition~\ref{def:mmc}.
\end{lemma}

\begin{proof}
We first fix a layer $i$ which is $\beta$-contributing. Let $\mathcal{U}(v)$ be the vector $\L(v)$ after removing buckets $j = 0,\dots,i$. %
By Lemma~\ref{lem:flat_median_lem}, there is an absolute constant $\lambda_1>0$ such that
\begin{align*}
     \|\L_i(v)\|_{\sym}
    = & ~ \alpha^i\cdot \sqrt{b_i} \cdot l(\xi^{(b_i)})\\
    \le & ~ \lambda_1\cdot\alpha^i\cdot \sqrt{b_i}\cdot M_{l^{(b_i)}},
\end{align*}
and similarly
\begin{align*}
    \|\mathcal{U}(v)\|_{\sym}\ge\frac{\lambda_2\alpha^i}{\sqrt{\log (d)}}\cdot (\sum_{j=i+1}^P b_j)^{1/2} M_{l^{(\sum_{j=i+1}^P b_j)}}.
\end{align*}
With these two inequalities, we can have the following deduction.

First, we have 
\begin{align*}
\|\L_i(v)\|_{\sym} \ge \beta \cdot \|\L(v)\|_{\sym} \ge \beta \cdot \|\mathcal{U}(v)\|_{\sym}. 
\end{align*}

Second, we assume that $b_i<\sum_{j=i+1}^P b_j$, as otherwise we are done:
\begin{align*}
    b_i\geq \sum_{j=i+1}^P b_j \geq \frac{\lambda \beta^2}{\mmc(l)^2\log^2 (d)}\cdot\sum_{j=i+1}^P b_j.
\end{align*}
Then, by the monotonicity of the median (Lemma~\ref{lem:monoton_of_med}), we have
\begin{align*}
    M_{l^{(b_i)}} \le \lambda_3 \cdot \mmc(l) \cdot \sqrt{\log (d)}\cdot M_{l^{(\sum_{j=i+1}^P b_j)}}
\end{align*}
for some absolute constant $\lambda_3 > 0$. 

Putting it all together, we get
\begin{align*}
    \beta \cdot \frac{\lambda_2\alpha^i}{\sqrt{\log (d)}} \cdot (\sum_{j=i+1}^P b_j)^{1/2} \le \lambda_1\cdot \alpha^i \sqrt{b_i}\cdot \lambda_3 \cdot \mmc(l) \cdot \sqrt{\log (d)}.
\end{align*}

Therefore, we finish the proof.
\end{proof}

\begin{lemma}[Importance of contributing layers (Part 2)%
]
\label{lem:right_beta_contributing}
For a symmetric $l$, let $\mmc(l)$ be defined as Definition~\ref{def:mmc}. %
If layer $i \in [P]$ is $\beta$-contributing, then there is an absolute constant $\lambda>0$ such that 
\begin{align*}
    b_i\alpha^{2i}\ge\frac{\lambda \beta^2}{\mmc(l)^2\cdot \log_\alpha (n) \cdot \log^2 (n)}\cdot\sum_{j\in [i]}b_j \alpha^{2j}.
\end{align*}
\end{lemma}

\begin{proof}
We first fix a layer $i$ which is $\beta$-contributing, and let $h:=\arg\max_{j\le i} \sqrt{b_j}\alpha^j$. 
We consider the two different cases as follows.

First, if $b_i\ge b_h$ then the lemma follows obviously by 

\begin{align*}
    & ~ \sum_{j\in [i]}b_j\alpha^{2j}\\
    \le & ~ t\cdot b_h\cdot \alpha^{2h}\\
    \le & ~ O(\log_\alpha (d))\cdot b_i\cdot \alpha ^{2i}.
\end{align*}
The second case is when $b_i<b_h$. With Definition~\ref{def:contribut_layer} and Lemma~\ref{lem:flat_median_lem}, we have
\begin{align*}
    & ~ \lambda_1\cdot \alpha^i\cdot \sqrt{b_i} \cdot M_{l^{(b_i)}}\\
    \ge & ~ \|\L_i\|_{\sym}\\
    \ge & ~ \beta\cdot \|\L\|_{\sym}\\
    \ge & ~ \lambda_2 \cdot \beta \cdot \alpha^h \cdot \sqrt{\frac{b_h}{\log (d)}}\cdot M_{l^{(b_h)}},
\end{align*}
for some absolute constants $\lambda_1,\lambda_2>0$, where the first step follows from Lemma~\ref{lem:flat_median_lem}, the second step follows from Definition~\ref{def:contribut_layer}, and the last step follows from Lemma~\ref{lem:flat_median_lem}. 

following from monotonicity of the median (Lemma~\ref{lem:monoton_of_med}), we can plugging in $M_{l^{(b_i)}}\le \lambda_3\cdot\mmc(l)\cdot\sqrt{\log (d)}\cdot M_{l^{(b_h)}}$, for some absolute constant $\lambda_3>0$, so that we have
\begin{align*}
    \lambda_1\cdot\alpha^i\cdot\sqrt{b_i}\cdot M_{l^{(b_i)}} & ~ \ge \frac{\lambda_2\cdot \beta \cdot \sqrt{b_h}\cdot \alpha^h}{\sqrt{\log (d)}} \cdot \frac{M_{l^{(b_i)}}}{\lambda_3\cdot \mmc(l)\cdot\sqrt{\log (d)}},\\
    \sqrt{b_i}\cdot \alpha^i & ~ \ge \frac{\lambda_2\cdot \beta\cdot \sqrt{b_h} \cdot \alpha^h}{\lambda_1 \lambda_3 \cdot \mmc(l) \cdot\log (d) }.
\end{align*}
Square the above inequality and we can see that $b_h\cdot\alpha^{2h} \ge \frac{1}{O(\log_\alpha (d))}\cdot\sum_{j \in [i]}b_j \alpha^{2j}$.

Thus we complete the proof.
\end{proof}

\section{Formal Main Result and Algorithms}
\label{sec:data_structure}

In this section, we state the formal version of our main theorem and algorithms. Section~\ref{sec:ds_for_sym_norm_est} presents our main result: a data structure for distance estimation with symmetric norm. Section~\ref{subsec:sparse_rc_tool} introduces the sparse recovery tools for sketching.

\subsection{Formal Version of Our Main Result}
\label{sec:ds_for_sym_norm_est}

\begin{algorithm}[!ht]\caption{Data structure for symmetric norm estimation}
\label{alg:main_layer_vector}
    \begin{algorithmic}[1]
    \State {\bf data structure} \textsc{DistanceOnSymmetricNorm} \Comment{Theorem~\ref{thm:main_formal}}
     \State 
         \State  {\bf private:}
    \Procedure{LayerVectorApprox}{$\alpha,b_1,b_2,\dots,b_P,d$} \Comment{Lemma~\ref{lem:approx_with_layer}}%
        \State For each $i \in [P]$, let $\alpha^i \cdot {\bf 1}_{b_{i}} \in \R^{b_i}$ denote a vector that has length $b_i$ and every entry is $\alpha^i$
        \State $\L \gets ( \alpha^1 \cdot {\bf 1}_{b_1}, \cdots, \alpha^P \cdot {\bf 1}_{b_P} , 0, \ldots, 0) \in \mathbb{R}^{d}$ \Comment{Generate the layer vector}
        \State \Return $\|\L\|_{\sym}$ \Comment{Return the norm of the estimated layer vector}
    \EndProcedure
    \State {\bf end data structure}
    \end{algorithmic}
\end{algorithm}

\begin{algorithm}[!ht]\caption{Data structure for symmetric norm estimation: members, formal version of Algorithm~\ref{alg:main_members_init_informal}}\label{alg:main_members}
    \begin{algorithmic}[1]
    \State {\bf data structure} \textsc{DistanceOnSymmetricNorm} \Comment{Theorem~\ref{thm:main_formal}}
    \State {\bf members} 
        \State \hspace{4mm} $d,n \in \mathbb{N}_+$ \Comment{$n$ is the number of points, $d$ is dimension}
        \State \hspace{4mm} $X = \{x_i \in \mathbb{R}^d\}_{i=1}^n$ \Comment{Set of points being queried}
        \State \hspace{4mm} $L \in \mathbb{N}_+$ \Comment{number of layers we subsample}
        \State \hspace{4mm} $R \in \mathbb{N}_+$ \Comment{number of substreams in one layer}
        \State \hspace{4mm} $\epsilon$, $\delta$ 
        \State \hspace{4mm} $\beta$ \Comment{used to cut important layer}
        \State \hspace{4mm} $U \in \mathbb{N}_+$ \Comment{number of parallel processing} 
        \State \hspace{4mm} \textsc{BatchHeavyHitter} $\{S_{r,l,u}\}_{r\in[R], l\in[L], u\in [U]}$ \label{line:def_data_stucture}
        \State \hspace{4mm} $\{H_{r,l,u}\subset [d]\times \R 
        \}_{r\in[R], l\in[L], u\in [U]}$ \Comment{each set $H$ has a size of $2/\beta$, and is used to store the output of \textsc{BatchHeavyHitter} 
        }
        
        \State \hspace{4mm} $\gamma$ \Comment{parameter used when cutting layer vector}
        \State \hspace{4mm} $\mathrm{bmap} \in \{0,1\}^{R \times L \times U \times d}$  
        
        \State \hspace{4mm} $\ov{x}_{r,l,u} \in \R^{n \times d}$, for each $r \in [R], l \in [L], u \in [U]$ \Comment{Substreams} 
    \State {\bf end members} 
    \State {\bf end data structure}
        \end{algorithmic}
\end{algorithm}

\begin{algorithm}[!ht]\caption{Data structure for symmetric norm estimation: init, formal version of Algorithm~\ref{alg:main_members_init_informal}}\label{alg:main_init}
    \begin{algorithmic}[1]
    \State {\bf data structure} \textsc{DistanceOnSymmetricNorm} \Comment{Theorem~\ref{thm:main_formal}}
    \State {\bf public:}
     \Procedure{Init}{$ \{ x_1, \cdots, x_i\} \subset \R^d, n \in \mathbb{N}_+, d \in \mathbb{N}_+, \delta \in (0,0.1), \epsilon \in (0,0.1)$} \Comment{Lemma~\ref{lem:init_time}}
     
        \State $n \gets n$, $d \gets d$, $\delta \gets \delta$, $\epsilon \gets \epsilon$
        \For{$i=1 \to n$}
            \State $x_i \gets x_i$
        \EndFor
        \State $\epsilon_1 \gets O(\frac{\epsilon^2}{\log d})$ \Comment{We define this notation for purpose of analysis} 
        \State $L \gets \log(d)$, $R \gets \Theta( \epsilon_1^{-2}  \log(n/\delta)\log^2 d )$, $U \gets \lceil\log(n d^2 / \delta)\rceil$  
        \State $\beta \gets O(\frac{\epsilon^5}{\mmc(l)^2\log^5 d})$  
        \For{$r\in[R], l \in [L], u\in [U]$}
            \State $S_{r,l,u}.\textsc{Init}(\sqrt{\beta},n+2,d)$   \label{line:sketch_init}\Comment{Theorem~\ref{thm:batch_heavy_hitter}}  
        \EndFor
        \For{$r\in[R], u\in [U], j \in [d], l \in [L]$}
                    \State Draw $\xi \in [0,1]$  
                     \If{$\xi \in [0, 2^{-l}]$}
                        \State $\mathrm{bmap}[r,l,u,j] \gets 1$ \label{line:bmap_init_1} 
                    \Else 
                         \State $\mathrm{bmap}[r,l,u,j] \gets 0$ \label{line:bmap_init_2}
                    \EndIf
        \EndFor
        \For{$r\in[R], u\in [U] , i \in [n], j \in [d], l \in [L]$}
                    
            \If{$\mathrm{bmap}[r,l,u,j] = 1 $} 
                \State $S_{r,l,u}.\textsc{EncodeSingle}(i,j,x_{i,j},d)$ \Comment{Theorem~\ref{thm:batch_heavy_hitter}} 
                \label{line:sketch_init_encode}
                \State $[\ov{x}_{r,l,u}]_{i,j} \gets x_{i,j}$ \Comment{Create a copy of subvectors}
            \Else 
                \State $[\ov{x}_{r,l,u}]_{i,j} \gets 0$ 
            \EndIf
 
        \EndFor 
     \EndProcedure
    \State {\bf end data structure}
    \end{algorithmic}
\end{algorithm}

\begin{algorithm}[!ht]\caption{Data structure for symmetric norm estimation: update}\label{alg:main_update}
    \begin{algorithmic}[1]
    \State {\bf data structure} \textsc{DistanceOnSymmetricNorm} \Comment{Theorem~\ref{thm:main_formal}}
    \State {\bf public:}
     \Procedure{Update}{$i \in [n], z \in \R^d$} 
     \Comment{Lemma~\ref{lem:update_time}}%
         \State \Comment{You want to replace $x_i$ by $z$}
         \For{$r\in[R], u\in [U] , j \in [d], l \in [L]$}
                    \If{$\mathrm{bmap}[r,l,u] = 1 $}
                        \State $S_{r,l,u}.\textsc{EncodeSingle}(i,j,z_{j},d)$  \label{line:sketch_update} \Comment{Theorem~\ref{thm:batch_heavy_hitter}}
                        \State $[\ov{x}_{r,l,u}]_{i,j} \gets z_{j}$ \Comment{Create a copy of subvectors}
                    \Else 
                        \State $[\ov{x}_{r,l,u}]_{i,j} \gets 0$
                    \EndIf
        \EndFor 
        
     \EndProcedure
    \State {\bf end data structure}
    \end{algorithmic}
\end{algorithm}

\begin{algorithm}[!ht]\caption{Data structure for symmetric norm estimation: query, formal version of Algorithm~\ref{alg:main_query_informal}}\label{alg:main_query}
    \begin{algorithmic}[1]
    \small
    \State {\bf data structure} \textsc{DistanceOnSymmetricNorm} \Comment{Theorem~\ref{thm:main_formal}} %
    \Procedure{Query}{$q \in \mathbb{R}^d$} \Comment{Lemma~\ref{lem:query_time}, \ref{lem:query_correct}}
        \State $P \gets O(\log_\alpha(d))$ \Comment{$P$ denotes the number of non-empty layer sets}
        \For{$r \in [R], l \in [L], u \in [U]$}
            \State $S_{r,l,u}.\textsc{Encode}(n+1,q,d)$ \Comment{Generate Sketch for $q$} \label{line:q_sketch_gen}
        \EndFor
         \State $\xi \gets $ chosen uniformly at random from $[1/2, 1]$
         \State $\gamma \gets \Theta(\epsilon)$
         \State $\alpha \gets 1+\gamma \cdot \xi$
         \State $P \gets O(\log_\alpha (d))=O(\log(d)/\epsilon)$ \Comment{$P$ denotes the number of non-empty layer sets}
         \For{$i \in [n]$}
            \For{$r \in [R], l \in [L], u \in [U]$}
                \State $ S_{r,l,u}.\textsc{Subtract}(n+2,n+1,i)$ \label{line:subtract_op}
                \State $H_{r,l,u} \gets S_{r,l,u}.\textsc{Decode}(n+2,\sqrt{\beta}, d)$ \Comment{This can be done in $\frac{2}{\beta} \poly(\log d)$} \label{line:decode}
                \State \Comment{At this point $H_{r,l,u}$ is a list of index, the value of each index is reset to $0$ }
                \For{$k \in H_{r,l,u}$} \label{line:seletct_start}
                    \State $\mathrm{value} \gets [\ov{x}_{r,l,u}]_{i,k}$
                    \State $H_{r,l,u}[k] \gets \mathrm{value}$
                    \State $w \gets \lceil \log (\mathrm{value} )/\log(\alpha) \rceil$ 
                    \If{$\alpha^{w-1} \ge \mathrm{value} / (1+\epsilon)$}
                        \State $H_{r,l,u}\gets \mathsf{null}$
                        \State {\bf break}
                    \EndIf
                \EndFor \label{line:seletct_end}
                \If{$H_{r,l,u} \not= \mathsf{null}$}
                    \State $H_{r,l} \gets H_{r,l,u}$
                \EndIf
            \EndFor
            \For{$l \in [L],k \in [P]$}
                    \State $A^i_{l,k} \gets |\{k ~|~ \exists k \in H_{r,l}, \alpha^{k-1} < |H_{r, l}[k]| \le \alpha^{k}\}|$ \label{line:compute_size_of_layer_counts}
            \EndFor
            \For{$k\in [P]$} \label{line:layer_counts_compute_start}
                \State $q^i_k \gets \max_{l\in[L]}{\{l ~|~ A^i_{l,k} \ge \frac{R\log(1/\delta)}{100\log(d)}\}}$ \label{line:max_q}\Comment{Definition~\ref{def:qi_maximizer}}
                \State If $q^i_k$ does not exist, then $\hat{\eta}^i_k \gets 0$; Else  $\hat{\eta}^i_k \gets \frac{A_{q^i_k,k}}{R(1+\epsilon_1)}$  \label{line:eta_approx}
                \State If $\hat{\eta}^i_k = 0$ then $c^i_k \gets 0$; Else $c^i_k \gets \frac{\log(1-\hat{\eta}^i_k)}{1-w^{-q_k}}$ \label{line:layer_approx}

            \EndFor \label{line:layer_counts_compute_end}
            \State $\mathrm{dst}_i \gets \textsc{LayerVetcorApprox}(\alpha, c^i_1, c^i_2,\dots,c^i_P, d)$
            \For{$r \in [R], l \in [L], u \in [U]$}
                \State $\{H_{r,l,u}\} \gets \{0\}$ \Comment{Reset the sets to use for next point}
            \EndFor
         \EndFor
         \State \Return $\{\mathrm{dst}_i\}_{i \in [n]}$
    \EndProcedure
    \State {\bf end data structure}
    \end{algorithmic}
\end{algorithm}

Here we divide function $\textsc{Query}$ presented in Theorem~\ref{thm:main_informal} into two versions. One takes only the query point $q \in \R^d$ as input to ask all the distances. Another takes a query point $q \in \R^d$ and a set $\mathsf{S} \subset [n]$ to ask for the distance with the specific set of points. The latter can be viewed as a more general version of the former. In the former parts, we have proved the correctness of the query (Lemma~\ref{lem:query_correct}), and the running time of version for all points (Lemma~\ref{lem:query_time}). %
Now we state the both in the following theorem, and the running time analysis for the latter will be stated in Section~\ref{sec:more_detail_time}.

\begin{theorem}[Main result, formal version of Theorem~\ref{thm:main_informal}]\label{thm:main_formal}
There is a data structure (Algorithm~\ref{alg:main_members}, \ref{alg:main_init}, \ref{alg:main_query}, \ref{alg:main_update}) uses $O(\epsilon^{-9} n (d + \mmc(l)^2) \log^{14}(n d/\delta))$ spaces for the Online Approximate Adaptive Symmetric Norm Distance Estimation Problem (Definition ~\ref{def:main_problem}) with the following procedures:
\begin{itemize}
    \item \textsc{Init}$( \{x_1, x_2, \dots, x_n\}\subset \R^d, \epsilon \in (0,1), \delta \in (0,1))$: Given $n$ data points $\{x_1, x_2, \dots, x_n\}\subset \R^d$, an accuracy parameter $\epsilon$ and a failure probability $\delta$ as input, the data structure preprocesses in time $O(\epsilon^{-9}n (d + \mmc(l)^2) \log^{14}(n d/\delta))$.
    
    \item \textsc{UpdateX}$(z \in \R^{d}, i \in [n])$: Given an update vector $z \in \R^d$ and index $i \in [n]$, the \textsc{UpdateX} takes $z$ and $i$ as input and updates the data structure with the new $i$-th data point in $O(\epsilon^{-4}d \log^9(n d/\delta))$ time.

    \item \textsc{Query}$(q \in \R^d)$ (Querying all points): Given a query point $q \in \R^d$, the \textsc{Query} operation takes $q$ as input and approximately estimates the symmetric norm distances from $q$ to all the data points $\{x_1, x_2, \dots, x_n\}\subset \R^d$ in time 
    \begin{align*}
        O(\epsilon^{-9} (d + n \cdot \mmc(l)^2) \log^{14}(n d/\delta))
    \end{align*}
    i.e. it provides a set of estimates $\{\mathrm{dst}_i\}_{i=1}^n$ such that:
    \begin{align*}
       \forall i \in[n], (1-\epsilon)\|q-x_{i}\|_{\sym} \leq \mathrm{dst}_{i} \leq(1+\epsilon)\|q-x_{i}\|_{\sym} 
    \end{align*}
     with probability at least $1 -\delta$. 
    \item \textsc{Query}$(q \in \R^d, \mathsf{S} \subseteq [n])$ (Querying a specific set $\mathsf{S}$ of points). Given a query point $q \in \R^d$ and an index $i \in [n]$, the \textsc{QueryOne} operation takes $q$ and $i$ as input and approximately estimates the symmetric norm distances from $q$ to the $i$-th point $x_i \in \R^d$ in time 
    \begin{align*}
     O(\epsilon^{-9} (d + |\mathsf{S}| \cdot \mmc(l)^2) \log^{14}(n d/\delta))
    \end{align*}
    i.e. it provides a estimated distance $\mathrm{dst} \in \R^{{\mathsf{S}}}$ such that:
    \begin{align*}
       (1-\epsilon)\|q-x_{i}\|_{\sym} \leq \mathrm{dst}_i \leq(1+\epsilon)\|q-x_{i}\|_{\sym} , \forall i \in \mathsf{S}
    \end{align*}
     with probability at least $1 -\delta$.
     
    \item \textsc{EstPair}$(i,j  \in [n])$ Given indices $i, j \in [n]$, the \textsc{EstPair} operation takes $i$ and $j$ as input and approximately estimates the symmetric norm distances from $i$-th to the $j$-th point $x_i, x_j \in \R^d$ in time 
    \begin{align*}
        O(\epsilon^{-9} \cdot \mmc(l)^2 \log^{14}(n d/\delta))
    \end{align*}
    i.e. it provides a estimated distance $\mathrm{pair}$ such that:
    \begin{align*}
       (1-\epsilon)\|x_i-x_j\|_{\sym} \leq \mathrm{pair} \leq(1+\epsilon)\|x_i-x_j\|_{\sym} 
    \end{align*}
     with probability at least $1 -\delta$.
     
\end{itemize}
\end{theorem}

\begin{proof}
In Lemma~\ref{lem:init_time_formal}, Lemma~\ref{lem:update_time_formal}, Lemma~\ref{lem:query_one_time} and Lemma~\ref{lem:query_pair_time} we analyze the running time for \textsc{Init}, \textsc{Update} and \textsc{Query} and \textsc{EstPair} respectively.

Lemma~\ref{lem:storage_formal} shows the space complexity.

In Lemma~\ref{lem:query_correct} we give the correctness of \textsc{Query}, and correctness for \textsc{EstPair} follows directly.

Thus, putting them all together, we prove the Theorem.

\end{proof}

\subsection{Sparse Recovery tools}
\label{subsec:sparse_rc_tool}

We start with describing a data structure problem
\begin{definition}[Batch Heavy Hitter]\label{def:bactch_heavy_hitter}
Given an $n \times d$ matrix, the goal is to design a data structure that supports the following operations:
\begin{itemize}
    \item \textsc{Init}{$(\epsilon \in (0,0.1), n, d)$}. Create a set of Random Hash functions and all the $n$ copies of sketches share the same hash functions. 
    \item \textsc{Encode}{$(i \in [n], z \in \R^d, d)$}. This step encodes $z$ into $i$-th sketched location and store a size ${\cal S}_{\mathrm{space}}$ linear sketch. 
    \item \textsc{EncodeSingle}{$(i \in [n], j\in [d], z \in \R, d)$}. This step updates one sparse vector $e_j z \in \R^d$ into $i$-th sketched location. 
    \item \textsc{Subtract}$(i,j,l \in [n])$. This function updates the sketch at $i$-th location by $j$-th sketch minus $l$-th sketch.
    \item \textsc{Decode}{$(i\in [n], \epsilon \in (0,0.1), d)$}. This function returns a set $L \subseteq [d]$ of size $|L| = O(\epsilon^{-2})$ containing all $\epsilon$-heavy hitters $i\in [n]$ under $\ell_p$. Here we say $i$ is an $\epsilon$-heavy hitter under $\ell_2$ if $|x_i| \geq \epsilon \cdot \| x_{ \ov{ [\epsilon^{-2}] } } \|_2$ where $x_{\ov{[k]}}$ denotes the vector $x$ with the largest $k$ entries (in absolute value) set to zero. Note that the number of heavy hitters never exceeds $2/\epsilon^2$.
\end{itemize}
\end{definition}

The existing work \cite{knpw11,p13} implies the following result. %
However their proofs are very decent and complicated. We provide another data structure in Section~\ref{sec:sparse_recv_tool} that significantly simplifies the analysis (by only paying some extra log factors). We believe it is of independent interest.
\begin{theorem}\label{thm:batch_heavy_hitter}
There is (linear sketch) data structure \textsc{BatchHeavyHitter}$(\epsilon,n,d)$ that uses ${\cal S}_{\mathrm{space}}$ space that support the following operations:
\begin{itemize}
    \item \textsc{Init}{$(\epsilon \in (0,0.1), n, d)$}. This step takes ${\cal T}_{\mathrm{init}}(\epsilon,n,d)$ time.
    \item \textsc{Encode}{$(i \in [n], z \in \R^d, d)$}. This step takes ${\cal T}_{\mathrm{encode}}(d)$ time. 
    \item \textsc{EncodeSingle}{$(i \in [n], j\in [d], z \in \R, d)$}. This step takes ${\cal T}_{\mathrm{encodesingle}}(d)$ time.
    \item \textsc{Subtract}$(i,j,l \in [n])$. This step takes ${\cal T}_{\mathrm{encodesingle}}(d)$ time.
    \item \textsc{Decode}{$(i\in [n], \epsilon \in (0,0.1), d)$}. This step takes ${\cal T}_{\mathrm{decode}}(\epsilon,d)$ time.
\end{itemize}
The running time of function can be summarize as 
\begin{itemize}
    \item ${\cal S}_{\mathrm{space}}(\epsilon,d) = n \cdot O( \epsilon^{-2} \log^2 d )$
    \item ${\cal T}_{\mathrm{init}}(\epsilon,n,d) = n \cdot O( \epsilon^{-2} \log^2 d )$
    \item ${\cal T}_{\mathrm{encode}}(d) = O( d \log^2(d) )$
    \item ${\cal T}_{\mathrm{encodesingle}}(d) = O( \log^2(d) )$
    \item ${\cal T}_{\mathrm{subtract}}(\epsilon,d) = O(\epsilon^{-2} \log^2 d)$
    \item ${\cal T}_{\mathrm{decode}}(\epsilon, d) = O( \epsilon^{-2} \log^2 d )$
\end{itemize}
Note that the succeed probability is at least $0.99$.
\end{theorem}
We remark that, to boost the probability from constant to $1-1/\poly(nd)$ we just need to pay an extra $\log(nd)$ factor.

\begin{algorithm}[!ht]\caption{Data structure for symmetric norm estimation: query set}\label{alg:main_query_one}
    \begin{algorithmic}[1]
    \small
    \State {\bf data structure} \textsc{DistanceOnSymmetricNorm}\Comment{Theorem~\ref{thm:main_formal}}
    \Procedure{Query}{$q \in \mathbb{R}^d$, $\mathsf{S} \subseteq [n]$} \Comment{Lemma~\ref{lem:query_time}, \ref{lem:query_correct}} 
        \State $P \gets O(\log_\alpha(d))$ \Comment{$P$ denotes the number of non-empty layer sets}
        \For{$r \in [R], l \in [L], u \in [U]$}
            \State $S_{r,l,u}.\textsc{Encode}(n+1,q,d)$ \Comment{Generate Sketch for $q$} \label{line:q_sketch_gen_query_one}
        \EndFor
         \State $\xi \gets $ chosen uniformly at random from $[1/2, 1]$
         \State $\gamma \gets \Theta(\epsilon)$
         \State $\alpha \gets 1+\gamma \cdot \xi$
         \State $P \gets O(\log_\alpha (d))=O(\log(d)/\epsilon)$ \Comment{$P$ denotes the number of non-empty layer sets}
        \For{$i \in \mathsf{S}$}
            \For{$r \in [R], l \in [L], u \in [U]$}
                \State $ S_{r,l,u}.\textsc{Subtract}(n+2,n+1,i)$ \label{line:subtract_op_2}
                \State $H_{r,l,u} \gets S_{r,l,u}.\textsc{Decode}(n+2,\sqrt{\beta}, d)$ \Comment{This can be done in $\frac{2}{\beta} \poly(\log d)$} \label{line:decode_2}
                \State \Comment{At this point $H_{r,l,u}$ is a list of index, the value of each index is reset to $0$ }
                \For{$k \in H_{r,l,u}$} \label{line:seletct_start_2}
                    \State $\mathrm{value} \gets [\ov{x}_{r,l,u}]_{i,k}$
                    \State $H_{r,l,u}[k] \gets \mathrm{value}$
                    \State $w \gets \lceil \log (\mathrm{value} )/\log(\alpha) \rceil$ 
                    \If{$\alpha^{w-1} \ge \mathrm{value} / (1+\epsilon)$}
                        \State $H_{r,l,u}\gets \mathsf{null}$
                        \State {\bf break}
                    \EndIf
                \EndFor \label{line:seletct_end_2}
                \If{$H_{r,l,u} \not= \mathsf{null}$}
                    \State $H_{r,l} \gets H_{r,l,u}$
                \EndIf
            \EndFor
            \For{$l \in [L],k \in [P]$}
                    \State $A^i_{l,k} \gets |\{k ~|~ \exists k \in H_{r,l}, \alpha^{k-1} < |H_{r, l}[k]| \le \alpha^{k}\}|$ \label{line:compute_size_of_layer_counts_2}
            \EndFor
            \For{$k\in [P]$} \label{line:layer_counts_compute_start_2}
                \State $q^i_k \gets \max_{l\in[L]}{\{l ~|~ A^i_{l,k} \ge \frac{R\log(1/\delta)}{100\log(d)}\}}$ \label{line:max_q_2}
                \State If $q^i_k$ does not exist, then $\hat{\eta}^i_k \gets 0$; Else  $\hat{\eta}^i_k \gets \frac{A_{q^i_k,k}}{R(1+\epsilon_2)}$  \label{line:eta_approx_2}
                \State If $\hat{\eta}^i_k = 0$ then $c^i_k \gets 0$; Else $c^i_k \gets \frac{\log(1-\hat{\eta}^i_k)}{1-w^{-q_k}}$ \label{line:layer_approx_2}

            \EndFor \label{line:layer_counts_compute_end_2}
            \State $\mathrm{dst}_i \gets \textsc{LayerVetcorApprox}(\alpha, c^i_1, c^i_2,\dots,c^i_P, d)$
            \For{$r \in [R], l \in [L], u \in [U]$}
                \State $\{H_{r,l,u}\} \gets \{0\}$ \Comment{Reset the sets to use for next point}
            \EndFor
         \EndFor
         \State \Return $\{\mathrm{dst}_i\}_{i \in \mathsf{S}}$
    \EndProcedure
    \State {\bf end data structure}
    \end{algorithmic}
\end{algorithm}

\begin{algorithm}[!ht]\caption{Data structure for symmetric norm estimation: query pair}\label{alg:main_query_pair}
    \begin{algorithmic}[1]
    \State {\bf data structure} \textsc{DistanceOnSymmetricNorm}\Comment{Theorem~\ref{thm:main_formal}}
    \Procedure{EstPair}{$i \in [n]$, $j \in [n]$} \Comment{Lemma~\ref{lem:query_time}, \ref{lem:query_correct}}
        \State $P \gets O(\log_\alpha(d))$ \Comment{$P$ denotes the number of non-empty layer sets}
         \State $\xi \gets $ chosen uniformly at random from $[1/2, 1]$
         \State $\gamma \gets \Theta(\epsilon)$
         \State $\alpha \gets 1+\gamma \cdot \xi$
         \State $P \gets O(\log_\alpha (d))=O(\log(d)/\epsilon)$ \Comment{$P$ denotes the number of non-empty layer sets}
            \For{$r \in [R], l \in [L], u \in [U]$}
                \State $ S_{r,l,u}.\textsc{Subtract}(n+2,j,i)$ \label{line:subtract_op_pair}
                \State $H_{r,l,u} \gets S_{r,l,u}.\textsc{Decode}(n+2,\sqrt{\beta}, d)$ \Comment{This can be done in $\frac{2}{\beta} \poly(\log d)$} \label{line:decode_pair}
                \State \Comment{At this point $H_{r,l,u}$ is a list of index, the value of each index is reset to $0$ }
                \For{$k \in H_{r,l,u}$} \label{line:seletct_start_pair}
                    \State $\mathrm{value} \gets [\ov{x}_{r,l,u}]_{i,k}$
                    \State $H_{r,l,u}[k] \gets \mathrm{value}$
                    \State $w \gets \lceil \log (\mathrm{value} )/\log(\alpha) \rceil$ 
                    \If{$\alpha^{w-1} \ge \mathrm{value} / (1+\epsilon)$}
                        \State $H_{r,l,u}\gets \mathsf{null}$
                        \State {\bf break}
                    \EndIf
                \EndFor \label{line:seletct_end_pair}
                \If{$H_{r,l,u} \not= \mathsf{null}$}
                    \State $H_{r,l} \gets H_{r,l,u}$
                \EndIf
            \EndFor
            \For{$l \in [L],k \in [P]$}
                    \State $A_{l,k} \gets |\{k ~|~ \exists k \in H_{r,l}, \alpha^{k-1} < |H_{r, l}[k]| \le \alpha^{k}\}|$ \label{line:compute_size_of_layer_counts_pair}
            \EndFor
            \For{$k\in [P]$} \label{line:layer_counts_compute_start_pair}
                \State $q_k \gets \max_{l\in[L]}{ \{ l ~|~ A_{l,k} \ge \frac{R\log(1/\delta)}{100\log(d)} \}}$ \label{line:max_q_pair}
                \State If $q_k$ does not exist, then $\hat{\eta}_k \gets 0$; Else  $\hat{\eta}_k \gets \frac{A_{q_k,k}}{R(1+\epsilon_2)}$  \label{line:eta_approx_pair}
                \State If $\hat{\eta}_k = 0$ then $c_k \gets 0$; Else $c_k \gets \frac{\log(1-\hat{\eta}_k)}{1-w^{-q_k}}$ \label{line:layer_approx_pair}

            \EndFor \label{line:layer_counts_compute_end_pair}
            \State $\mathrm{dst} \gets \textsc{LayerVetcorApprox}(\alpha, c_1, c_2,\dots,c_P, d)$
         \State \Return $\mathrm{dst}$
    \EndProcedure
    \State {\bf end data structure}
    \end{algorithmic}
\end{algorithm}

\section{More Details of the Time Complexity}\label{sec:more_detail_time}

In this section, we analyze the running time of the general version of \textsc{Query}. Lemma~\ref{lem:query_time} is a special case of the following Lemma when $\mathsf{S} = [n]$.

\begin{lemma}[\textsc{Query} time for general version]
\label{lem:query_one_time}

Given a query point $q \in \R^d$ and a set $\mathsf{S} \subseteq [n]$, the procedure \textsc{Query} (Algorithm~\ref{alg:main_query_one}) runs in time 
\begin{align*}
     O(\epsilon^{-9} (d + |\mathsf{S}| \cdot \mmc(l)^2) \log^{14}(n d/\delta)).
\end{align*}
\end{lemma}

\begin{proof}
The \textsc{Query} operation for a stored vector (Algorithm~\ref{alg:main_query_one}) has the following two parts:
\begin{itemize}
    \item {\bf Part 1: }Line~\ref{line:q_sketch_gen_query_one} takes $O(R L U \cdot \T_{\mathrm{encode}})$ time to call \textsc{Encode} to generate sketches for $q$.
    \item {\bf Part 2: }For each $i \in \mathsf{S}$
    \begin{itemize}
        \item Line~\ref{line:subtract_op_2} takes $O(R L U \cdot \T_{\mathrm{subtract}})$ time to compute sketch of the difference between $q$ and $x_i$, and store the sketch at index of $n + 2$.
        \item Line~\ref{line:decode_2} takes $O(R L U \cdot \T_{\mathrm{decode}})$ time to decode the \textsc{BatchHeavyHitter} and get estimated heavy hitters of $q - x_i$.
        \item Line~\ref{line:seletct_start_2} to Line~\ref{line:seletct_end_2} takes $O(R L U\cdot 2/\beta)$ time to analyze the \textsc{BatchHeavyHitter} 
        and get the set of indices, where $2/\beta$ is the size of the set. 
        \item Line~\ref{line:compute_size_of_layer_counts_2} takes $O(L P \cdot 2/\beta)$ time to compute size of the layer sets cut by $\alpha$.
        \item Line~\ref{line:layer_counts_compute_start_2} to Line~\ref{line:layer_counts_compute_end_2} takes $O(P L)$ time to compute the estimation of each layer.
    \end{itemize}
    The total running time of this part is:
    \begin{align*}
            & ~ |\mathsf{S}| \cdot (O(R L U \cdot \T_{\mathrm{subtract}}) + O(R L U \cdot \T_{\mathrm{decode}}) + O(R L U\cdot 2/\beta) + O(L P \cdot 2/\beta) + O(LP)))\\
        =   & ~ O(|\mathsf{S}| \cdot L (R U (\T_{\mathrm{subtract}} + \T_{\mathrm{decode}} + \beta^{-1}) + P \beta^{-1}))
    \end{align*}
    time in total.
\end{itemize}
Taking these two parts together we have the total running time of the \textsc{Query} procedure:
\begin{align*}
        & ~ O(R L U \cdot \T_{\mathrm{encode}}) + O(|\mathsf{S}| \cdot L (R U (\T_{\mathrm{subtract}} + \T_{\mathrm{decode}} + \beta^{-1}) + P \beta^{-1})) \\
    =   & ~ O(R L U(\T_{\mathrm{encode}} + |\mathsf{S}| \cdot (\T_{\mathrm{subtract}} + \T_{\mathrm{decode}} + \beta^{-1})) + |\mathsf{S}| \cdot L P \beta^{-1})\\
    =   & ~ O(\epsilon^{-4} \log^6(n d/\delta)(d\log^2(d)\log(n d) + |\mathsf{S}| \cdot \beta^{-1}\log^2(d) \log(n d)) +  |\mathsf{S}| \cdot \epsilon^{-1} \log^2(d) \beta^{-1})\\
    =   & ~ O(\epsilon^{-9} (d + |\mathsf{S}| \cdot \mmc(l)^2) \log^{14}(n d/\delta))
\end{align*}
where the first step follows from the property of big $O$ notation, the second step follows from the definition of $R,L,U,\T_{\mathrm{encode}}, \T_{\mathrm{encode}}, \T_{\mathrm{subtract}}, \T_{\mathrm{decode}}$ (Theorem~\ref{thm:batch_heavy_hitter})
$,P$, the third step follows from merging the terms.

Thus, we complete the proof.

\end{proof}

\begin{lemma}[\textsc{EstPair} time]
\label{lem:query_pair_time}

Given a query point $q \in \R^d$, the procedure \textsc{EstPair} (Algorithm~\ref{alg:main_query_pair}) runs in time 
\begin{align*}
    O(\epsilon^{-9} \cdot \mmc(l)^2 \log^{14}(n d/\delta)).
\end{align*}
\end{lemma}

\begin{proof}
The \textsc{EstPair} operation (Algorithm~\ref{alg:main_query_pair}) has the following two parts:
\begin{itemize}
    \item Line~\ref{line:subtract_op_pair} takes $O(R L U \cdot \T_{\mathrm{subtract}})$ time to compute sketch of the difference between %
    $x_i$ and $x_j$, and store the sketch at index of $n + 2$.
    \item Line~\ref{line:decode_pair} takes $O(R L U \cdot \T_{\mathrm{decode}})$ time to decode the \textsc{BatchHeavyHitter} and get estimated heavy hitters of $x_i - x_j$.
    \item Line~\ref{line:seletct_start_pair} to Line~\ref{line:seletct_end_pair} takes $O(R L U\cdot 2/\beta)$ time to analyze the \textsc{BatchHeavyHitter} 
    and get the set of indices, where $2/\beta$ is the size of the set. 
    \item Line~\ref{line:compute_size_of_layer_counts_pair} takes $O(L P \cdot 2/\beta)$ time to compute size of the layer sets cut by $\alpha$.
    \item Line~\ref{line:layer_counts_compute_start_pair} to Line~\ref{line:layer_counts_compute_end_pair} takes $O(P L)$ time to compute the estimation of each layer.
\end{itemize}

The total running time is:
\begin{align*}
        & ~ O(R L U \cdot \T_{\mathrm{subtract}}) + O(R L U \cdot \T_{\mathrm{decode}}) + O(R L U\cdot 2/\beta) + O(L P \cdot 2/\beta) + O(LP)\\
    =   & ~ O(R L U(\T_{\mathrm{subtract}} + \T_{\mathrm{decode}} + \beta^{-1}) + L P \beta^{-1})\\
    =   & ~ O(\epsilon^{-4} \log^6(n d/\delta)(\beta^{-1}\log^2(d) \log(n d) +  \epsilon^{-1} \log^2(d) \beta^{-1} \log(n d))\\
    =   & ~ O(\epsilon^{-9} \cdot \mmc(l)^2 \log^{14}(n d/\delta))
\end{align*}
time in total, where the first step follows from the property of big $O$ notation, the second step follows from the definition of $R,L,U,\T_{\mathrm{encode}}, \T_{\mathrm{encode}}, \T_{\mathrm{subtract}}, \T_{\mathrm{decode}}$ (Theorem~\ref{thm:batch_heavy_hitter})
$,P$, the third step follows from merging the terms.

Thus, we complete the proof.

\end{proof}

\begin{lemma}[\textsc{Init} time, formal version of Lemma~\ref{lem:init_time}]
\label{lem:init_time_formal}
    Given data points $\{x_1,x_2,\dots,x_n\} \subset \R^d$, an accuracy parameter $\epsilon>0$, and a failure probability $\delta>0$ as input, the procedure \textsc{init} (Algorithm~\ref{alg:main_init}) runs in time
    \begin{align*}
        O(\epsilon^{-9}n (d + \mmc(l)^2) \log^{14}(n d/\delta)).
    \end{align*}
\end{lemma}
\begin{proof}
The \textsc{Init} time includes these parts:
\begin{itemize}
    \item Line~\ref{line:sketch_init} takes $O(R L U \cdot {\cal T}_{\mathrm{init}}(\sqrt{\beta},n,d))$ to initialize sketches %
    \item Line~\ref{line:bmap_init_1} to Line~\ref{line:bmap_init_2} takes $O(R U d L)$ to generate the \rm{bmap};
    \item Line~\ref{line:sketch_init_encode} takes $O(n d R U L \cdot \T_{\mathrm{encodesingle}}(d))$ to generate sketches %
\end{itemize}

By Theorems~\ref{thm:batch_heavy_hitter}, we have
\begin{itemize}
    \item $\T_{\mathrm{init}}(\sqrt{\beta}, n, d)) = n \cdot O(\beta^{-1} \log^2 (d)  \log(n d)) = O(n \cdot \mmc(l)^2 \log^7(d)\log(n d) \epsilon^{-5} )$,
    \item  $\T_{\mathrm{encodesingle}}(d) = O(\log^2(d)\log(nd))$.
\end{itemize}

Adding them together we got the time of
\begin{align*}
        & ~ O(R L U {\cal T}_{\mathrm{init}}(\sqrt{\beta},n,d)) + O(R U d L) + O(n d R U L \cdot \T_{\mathrm{encodesingle}}(d)) \\
    =   & ~ O(R L U({\cal T}_{\mathrm{init}}(\sqrt{\beta},n,d) + n d \cdot  \T_{\mathrm{encodesingle}}(d)))\\
    =   & ~ O(\epsilon^{-4} \log(d/\delta) \log^4(d) \cdot \log(d) \cdot \log(n d^2/\delta) \log(n d) (n \cdot \mmc(l)^2 \log^7(d) \epsilon^{-5} + n d \log^2(d)))\\
    =   & ~ O(\epsilon^{-9}n (d + \mmc(l)^2) \log^{14}(n d/\delta)),
\end{align*}
where the first step follows from merging the terms, the second step follows from the definition of $R,L,U,\T_{\mathrm{encodesingle}}(d),{\cal T}_{\mathrm{init}}$, the third step follows from merging the terms.

Thus, we complete the proof.
\end{proof}

\begin{lemma}[\textsc{Update} time, formal version of Lemma~\ref{lem:update_time}]
\label{lem:update_time_formal}
    Given a new data point $z \in \R^d$, and an index $i$ where it should replace the original data point $x_i \in \R^d$. The procedure \textsc{Update} (Algorithm~\ref{alg:main_update}) runs in time
    \begin{align*}
        O(\epsilon^{-4}d \log^9(n d/\delta)).
    \end{align*}
\end{lemma}

\begin{proof} 
The \textsc{Update} operation calls \textsc{BatchHeavyHitter.Encode} for $RLU$ times, so it has the time of 
\begin{align*}
         O(R L U \cdot \T_{\mathrm{encode}}(d)) 
    =   & ~ O(\epsilon^{-4} \log(n/\delta) \log^4(d) \cdot \log(d) \cdot \log(n d^2/\delta) \cdot d \log^2(d) \log(n d))\\
    =   & ~ O(\epsilon^{-4}d \log^9(n d/\delta))
\end{align*}
where the first step follows from the definition of $R,L,U,\T_{\mathrm{encode}}(d)$, the second step follows from
\begin{align*}
        & ~ \log(n/\delta)\log^4(d)  \log(d) \log(n d^2/\delta) \log^2 (d) \log(n d) \\
    =   & ~ (\log(n/\delta)) (\log^7 d) (2\log d + \log(n/\delta))\log(n d)  \\
    =   & ~ O(\log^9(n d/\delta))
\end{align*}
Thus, we complete the proof.
\end{proof}

\section{More Details of the Correctness Proofs}
\label{sec:detail_correctness}

In this section, we give the complete proofs of the correctness of our algorithms. In Section~\ref{subsec:correctness_Layer_Size_Estimation} we state and prove the main result of this section, using the technical lemmas in the following subsections. In Section~\ref{subsec:detectability}, we define the trackable layers and show the connection with important layers (Definition~\ref{def:important_layers}). In Section~\ref{subsec:probability_analysis}, we analyze the sample probability and track probability of a layer. In Section~\ref{subsec:probability_leads_to_layer_sizes}, we show that a good estimation of track probability %
implies a good approximation of the layer size, which completes the proof of the correctness.

\subsection{Correctness of Layer Size Estimation}
\label{subsec:correctness_Layer_Size_Estimation}

We first show that, the estimation of layer sizes output by our data structure are good to approximate the exact values.

\begin{lemma}[Correctness of layer size approximation]
\label{lem:layer_count_approx}
    We first show that, the layer sizes $c^i_1, c^i_2, \dots, c^i_P$ our data structure return satisfy
    \begin{itemize}
        \item for all $k \in P$, $c^i_k \le b^i_k$;
        \item if $k$ is a $\beta$-important layer (Definition~\ref{def:important_layers}) of $q-x_i$, then $c^i_k \ge (1 - \epsilon_1)b^i_k$,
    \end{itemize}
    with probability at least $1-\delta$, where the $b^i_k$ is the ground truth $k$-th layer size of $q-x_i$.
\end{lemma}
\begin{proof}
    We first define two events as
    \begin{itemize}
        \item $E_1$: for all important layers $k \in [P]$, $q_k$ is well defined;
        \item $E_2$: for all $k \in [P]$, if $\hat{\eta}_k > 0$ then $(1-O(\epsilon))\eta'_{k,q_k} \le \hat{\eta}_k \le \eta'_{k,q_k}$.
    \end{itemize}
    With Lemma~\ref{lem:eta_approx} and Lemma~\ref{lem:eta_approx_prob}, we have that
    \begin{align*}
        \Pr[E_1 \cap E_2] \ge 1-\delta^{O(\log(d))}.
    \end{align*}
    When the output of every \textsc{BatchHeavyHitter} is correct, if follows from Lemma~\ref{lem:approx_B_with_H}, Lemma~\ref{lem:eta_approx} and Lemma~\ref{lem:probability_to_sizeof_layers} the algorithm outputs an approximation to the layer vector meeting the two criteria.
    
    The \textsc{BatchHeavyHitter} is used a total of $LR$ times, each with error probability at most $\delta/\d$. By a union bound over the layers, the failure probability is at most $(\poly(\log d)) \cdot \delta /d = o(\delta)$. Therefore, the total failure probability of the algorithm is at most $1-o(\delta)$.
    
    Thus, we complete the proof.
\end{proof}

\subsection{Trackability of Layers}
\label{subsec:detectability}

\begin{definition}[Trackability of layers]
\label{def:detectability_of_layer}
    A layer $k \in [P]$ of a vector $x$ is $\beta$-trackable, if
    \begin{align*}
        \alpha^{2k} \ge \beta \cdot \| x_{ \ov{ [\beta^{-1}] } } \|_2^2
    \end{align*}
    where $x_{\ov{[\kappa]}}$ is defined as Definition~\ref{def:tail}. %
\end{definition}

\begin{lemma}[Importance and Trackability]
\label{lem:import_to_detect}
Let $\alpha$ be some parameter such that $\alpha \in [0,2]$. 
    Suppose subvector $\wt{x}$ %
    is obtained by subsampling the original vector with probability $p$.

    If $k \in [P]$ %
    is a $\beta$-important layer, then for any $\lambda > P$, %
    with probability at least $1-P\exp(-\frac{\lambda p b_k}{P \beta})$,  
    layer $k$ is $\frac{\beta}{\lambda p b_k}$-trackable.

    In particular, if $p b_k=O(1)$, then with probability at least $1-P\exp(-\Omega(\frac{\lambda}{P\beta}))$,
    layer $k$ is $\frac{\beta}{\lambda}$-trackable.
\end{lemma}

\begin{proof}
    Let $(\zeta_0, \zeta_1,\dots)$ be the new level sizes of the subvector $\wt{x}$ %
    we sampled. Thus, for $k \in [P]$, we have 
    \begin{align*}
        \E[\zeta_k] = p \cdot b_k.
    \end{align*}
    By the definition of important layer (Definition~\ref{def:important_layers}) we have
    \begin{align*}
        \E[\zeta_k] \ge \beta \cdot \E\Big[\sum_{j = k+1}^P \zeta_j\Big]
    \end{align*}
    and
    \begin{align*}
        \E[\zeta_k \cdot \alpha^{2k}] \ge \beta \cdot \E\Big[\sum_{j \in [k]} \zeta_j \cdot \alpha^{2j}\Big].
    \end{align*}
    To have layer $k$ trackable, it has to be in the top $\lambda p b_k/\beta$ elements, so that we have $\sum_{j = k+1}^P\zeta_j \le \lambda p b_k/\beta$. By Chernoff bound (Lemma~\ref{lem:Chernoff_bound}), we can know the probability that this event does not happen is
    \begin{align*}
        \Pr\Big[\sum_{j = k+1}^P b_j> \frac{\lambda p b_k}{\beta}\Big] \le \exp(-\Omega(\lambda p b_k/\beta)).
    \end{align*}
    On the other hand, for level $k$ to be trackable $\alpha^{2k} \ge \frac{\beta}{\lambda p b_k}\sum_{j \in [k]}\zeta_j \alpha^{2j}$. 
    
    Thus, the complement occurs with probability
    \begin{align*}
            \Pr\Big[\frac{\beta}{\lambda p b_k}\sum_{j \in [k]}\zeta_j \alpha^{2j} > \alpha^{2k}\Big]
        \le & ~ \Pr\Big[\exists j, \zeta_j \alpha^{2j} \ge \frac{\lambda p b_k \alpha^{2k}}{P \beta}\Big] \\
        \le & ~ \sum_{j \in [P]}\Pr\Big[\zeta_j \alpha^{2j} \ge \frac{\lambda p b_k \alpha^{2k}}{P \beta}\Big].
    \end{align*}
    By Chernoff bound (Lemma~\ref{lem:Chernoff_bound}) and the fact that $\E[\zeta_j\alpha^{2j}] \le p b_k \alpha^{2k}$, we have
    \begin{align*}
        \Pr\Big[\frac{\beta}{\lambda p b_k}\sum_{j \in [k]}\zeta_j \alpha^{2j} > \alpha^{2k}\Big] 
    \le & ~ P \cdot \exp\Big(-\Omega(\frac{\lambda p b_k \alpha^{2k}}{\alpha^{2j} P \beta})\Big)\\
    \le & ~ P \cdot \exp\Big(-\Omega(\frac{\lambda p b_k}{\beta})\Big).
    \end{align*}
    Thus we complete the proof.
\end{proof}

\subsection{Probability Analysis}
\label{subsec:probability_analysis}

We first define  %
\begin{definition}
For each $k \in [p]$, $l \in \R^+$, %
we define
$\eta_{k,l}:=1 - (1-p_l)^{b_k}$ to be the probability that at least one element from $B_k$ is sampled with the sampling probability of $p_l=2^{-l}$. 
\end{definition}
Set $\lambda = \Theta(P \log (1/\delta))$. Let $\eta^*_{k,l}$ be the probability that an element from $B_k$ is contained in $H_{1,l}$, such that, for $H_{r,l}$ with any other $r$, the probability is the same as $\eta^*_{k,l}$.

\begin{lemma}[Sample Probability and Track Probability]
\label{lem:approx_B_with_H}
    For any layer $k \in [P]$ we have $\eta^*_{k,l} < \eta_{k,l}$. Let $\delta$ and $\epsilon$ denote the two parameters such that  $ 0 < \delta < \epsilon < 1$. If layer $k$ is a $\beta$-important layer and $p_l b_k = O(1)$, then 
    \begin{align*}
        \eta^*_{k,l} \ge (1-\Theta(\epsilon)) \cdot \eta_{k,l}.
    \end{align*}
\end{lemma}
\begin{proof}

    If one is in $H_{r,l}$, it has to be sampled, so we have $\eta^*_{k,l} \le \eta_{k,l}$. On the other hand, using Lemma~\ref{lem:import_to_detect}, with probability at least $1 - t\exp(-\Omega(\frac{\lambda p b_k}{t \beta}))$, layer $k$ is $\frac{\beta}{\lambda p b_k}$-trackable.
    
    We have
    \begin{align*}
        \frac{\beta}{\lambda p b_k}=\Theta(\frac{\beta}{P \log(1/\delta)})
    \end{align*}
    where the last step follows from definition of $\lambda$. 
    
    Thus with probability at least 
    \begin{align*} \eta_{k,l}(1-\Theta(\delta)) \ge \eta_{k,l}(1-O(\epsilon))
    \end{align*}
    
    an element from $B_k$ is sampled and the element is reported by \textsc{BatchHeavyHitter}.
    
    Thus we complete the proof.
\end{proof}

\begin{lemma}[Probability Approximation]
\label{lem:eta_approx}
    For $k\in [P]$, let $\hat{\eta}_k$ be defined as in Line~\ref{line:eta_approx} in Algorithm~\ref{alg:main_query}.
    With probability at least $1-\delta^{\Omega(\log d)}$, for all $k \in [P]$, if $\hat{\eta}_k \not= 0$ then 
    \begin{align*} 
    (1-O(\epsilon_1))\eta^*_{k, q_k} \le \hat{\eta}_k \le \eta^*_{k, q_k} .
    \end{align*}
\end{lemma}
\begin{proof}
    We define 
    \begin{align*}
        \gamma := \frac{R\log(1/\delta)}{100\log d}.
    \end{align*}
    
Recall the condition of Line~\ref{line:eta_approx}, if $\hat{\eta}_k \not= 0$, then we have
    \begin{align*}
        A_{q_k,k}\ge \gamma.
    \end{align*}
    For a fixed $k \in [P]$, %
    since the sampling process is independent for each $r \in [R]$, we can assume that
    \begin{align*}
        \E[A_{q_k,k}] = R\eta_{k,q_k} \ge \gamma.
    \end{align*}
    Otherwise, by Chernoff bound (Lemma~\ref{lem:Chernoff_bound}), we get that
    \begin{align*}
        \Pr[A_{q_k,k} \ge \gamma] = o(\delta^{\Omega(\log d)}),
    \end{align*}
    which implies that with probability at least $1-\delta^{\Omega(\log d)}$, $\hat{\eta}_k=0$ for all $k\in [P]$, and we are done.

    Thus, under this assumption, by Chernoff bound (Lemma~\ref{lem:Chernoff_bound}), we have
    \begin{align*}
        \Pr[|A_{q_k,k} - R \cdot \eta^*_{k,q_k}| \ge \epsilon R \eta_{k,q_k}] \le \exp(-\Omega(\epsilon^2\gamma)) = \delta^{\Omega(\log d)}.
    \end{align*}
    Since $P = \poly\log(d)$, by union bound over the $B_k$
    , the event
    \begin{align*}
        (1-\epsilon_1)\eta_{k,q_k} \le \frac{A_{q_k,k}}{R_k} \le (1+\epsilon_1)\eta_{k,q_k}
    \end{align*}
    holds for all $k \in [P]$ with probability at least $1-\delta^{\Omega(\log (d))}$. Since $\hat{\eta}_k=\frac{A_{k,q_k}}{R(1+\epsilon_1)}$ by definition, we get the desired bounds.
    
    Thus we complete the proof.
\end{proof}

\begin{definition}
\label{def:qi_maximizer}
    We define $q_k$ as
    \begin{align*}
        q_k := \max_{l \in [L]} \Big\{ l | A_{l,k} \ge \frac{R \log(1/\delta)}{100 \log(d)} \Big\}.
    \end{align*}
    We say that $q_k$ is well-defined if the set in the RHS of the above definition is non-empty.
\end{definition}

At Line~\ref{line:max_q} in Algorithm~\ref{alg:main_query}, we define the $q^i_k$ for $i$-th vector as the definition above.

\begin{lemma}[Maximizer Probability]
\label{lem:eta_approx_prob}
    If layer $k \in [P]$ is important (Line~\ref{line:max_q} in Algorithm~\ref{alg:main_query}), %
    then with probability at least $1-\delta^{\Omega(\log d)}$, the maximizer $q_k$ is well defined. %
\end{lemma} 
\begin{proof}
    Lemma~\ref{lem:approx_B_with_H} tells us that, when $p_k b_k = O(1)$, layer $k$ is at least $\Omega(\frac{\beta}{P\log(1/\delta)})$-trackable. On the other hand, if $P_k=2^{-l_0}=1/b_k$, we have $\eta_{k,l_0} = 1 - (1-p_k)^{b_k} \ge 1/e$. Thus we have
    \begin{align*}
        \E[A_{l_0,k}] \ge R/e,
    \end{align*}
    so by Chernoff bound (Lemma~\ref{lem:Chernoff_bound}),
    \begin{align*}
        \Pr[A_{l_0,k} \le & ~ \frac{R\log(1/\delta)}{\log d}] \le \exp(-\Omega(\frac{R\log(1/\delta)}{\log d})) \\
        \le & ~ \delta^{\Omega(\log d)}.
    \end{align*}
    Thus we show that, there exists one $q_k \ge l_0$ with probability at least $1 - \delta^{\Omega(\log n)}$.

    Since there are at most $P = \poly\log (d)$ important layers, with probability at least $1 - \delta^{\Omega(\log n)}$, the corresponding value of $q_k$ is well defined for all important layers.
    
    Thus we complete the proof.
\end{proof}

\subsection{From Probability Estimation to Layer Size Approximation}
\label{subsec:probability_leads_to_layer_sizes}
 
The following lemma shows that, if we have a sharp estimate for the track probability of a layer, then we can obtain a good approximation for its size.  

\begin{lemma}[Track probability implies layer size approximation]
\label{lem:probability_to_sizeof_layers}
    Suppose $q_k \ge 1$, $\epsilon \in (0,1/2)$, %
    and $d$ is sufficiently large. Define $\epsilon_1 := O(\epsilon^2/\log(d))$. We have for $k \in [P]$,
    \begin{itemize}
        \item Part 1. If $\hat{\eta}_k \le \eta_{k,q_k}$ then $c_k \le b_k$.
        \item Part 2. If $\hat{\eta}_i \ge (1 - \epsilon_1)\eta_{k, q_k}$ then $c_k \ge (1-O(\epsilon_1))b_k$.
    \end{itemize}
\end{lemma}
\begin{proof}

    We know that
    \begin{align*}
        b_k = \frac{\log(1-\eta_{k, q_k})}{\log(1 - 2^{-q_k})},
    \end{align*}
    which is a increasing function of $\eta_{k,q_k}$. 
  
   {\bf Part 1.}
   If $\hat{\eta}_k \le \eta_{k,q_k}$, then $c_k \le b_k$.

  {\bf Part 2.} If $\hat{\eta}_k \ge (1-O(\epsilon))\eta_{k, q_k}$ we have
    \begin{align*}
        c_k \ge b_k + \frac{\epsilon_1 \eta_{k,q_k}}{(1-\eta_{k, q_k}) \log(1-2^{-q_k})} \ge b_k - O(\epsilon) b_k.
    \end{align*}
    Thus we complete the proof.
\end{proof}

\section{Space Complexity}
\label{sec:space}

In this section, we prove the space complexity of our data structure.

\begin{lemma}[Space complexity of our data structure, formal version of Lemma~\ref{lem:storage}]
\label{lem:storage_formal}
Our data structure (Algorithm~\ref{alg:main_members} and \ref{alg:main_init}) uses $O(\epsilon^{-9} n (d + \mmc(l)^2) \log^{14}(n d/\delta))$ space. %
\end{lemma}

\begin{proof}

First, we store the original data,
\begin{align*}
    \mathrm{space~for~} x = O(n d).
\end{align*}

Second, we store the sub stream/sample of orignal data
\begin{align*}
    \mathrm{space~for~} \ov{x} 
    =   & ~ O(R L U n d)\\
    =   & ~ O(\epsilon^{-4} \log(n/\delta) \log^4(d) \cdot \log(d) \cdot \log(n d^2/\delta) \cdot n d)\\
    =   & ~ O(\epsilon^{-4} n d \log^6(n d/\delta)).
\end{align*}

Our data structure holds a set $\{S_{r,l,u}\}_{r \in [R], l \in [L], u \in [U]}$ (Line~\ref{line:def_data_stucture}). Each $S_{r,l,u}$ has a size of $\mathcal{S}_{\mathrm{space}}(\sqrt{\beta}, d)=O(n \cdot \beta^{-1} \log^2(d) \log(n d))$, which uses the space of
\begin{align*}
   \mathrm{space~for~} S 
   =    & ~ O(R L U n \cdot \beta^{-1} \log^2(d) \log(nd))\\
   =    & ~ O(\epsilon^{-4} \log(n/\delta) \log^4(d) \cdot \log(d) \cdot \log(n d) \cdot \log(n d^2/\delta) \cdot n \cdot (\epsilon^{-5} \cdot \mmc(l)^2 \log^5(d)) \log^2(d))\\
   =    & ~ O(\epsilon^{-9} n \cdot \mmc(l)^2 \log^{14}(n d/\delta))
\end{align*}
where the first step follows from the definition of $R,L,U$, and the second step follows just simplifying the last step.

We hold a $\mathrm{bmap}$ (Line~\ref{line:bmap_init_1} and Line~\ref{line:bmap_init_2}) of size $O(R L U d)$, which uses the space of
\begin{align*}
\mathrm{space~for~bmap}    
    = & ~ O(R L U d) \\
    =   & ~ O(\epsilon^{-4} \log(n/\delta) \log^4(d) \cdot \log(d) \cdot \log(n d^2/\delta) \cdot d)\\
    =   & ~ O(\epsilon^{-4} d \log^6(n d/\delta)).
\end{align*}

In \textsc{Query}, we generate a set of sets $\{H_{r,l,u}\}_{r \in [R], l \in [L], u \in [U]}$, each of the sets has size of $O(\beta^{-1})$, so the whole set uses space of 

\begin{align*}
    \mathrm{space~for~}H
    =   & ~ O(R L U \cdot \beta^{-1})\\
    =   & ~ O(\epsilon^{-4} \log(n/\delta) \log^4(d) \cdot \log(d) \cdot \log(n d^2/\delta) \cdot \epsilon^{-5} \cdot \mmc(l)^2\log^5(d))\\
    =   & ~ O(\epsilon^{-9} \cdot \mmc(l)^2 \log^{11}(n d/\delta)).
\end{align*}
By putting them together, we have the total space is
\begin{align*}
        & ~ \mathrm{total~space} \\
     =  & ~ \mathrm{space~for~} x + \mathrm{space~for~} \ov{x} + \mathrm{space~for~} S + \mathrm{space~for~bmap} + \mathrm{space~for~}H\\
     =  & ~ O(n d) + O(\epsilon^{-4} n d \log^6(n d/\delta)) + O(\epsilon^{-9} n \cdot \mmc(l)^2 \log^{13}(n d/\delta))\\
     & ~  + O(\epsilon^{-4} d \log^6(n d/\delta)) + O(\epsilon^{-9} \cdot \mmc(l)^2 \log^{11}(n d/\delta))\\
     =  & ~ O(\epsilon^{-9} n (d + \mmc(l)^2) \log^{14}(n d/\delta)).
\end{align*}

Thus, we complete the proof.

\end{proof}

\section{Lower Bound From Previous Work}
\label{sec:lower_bound}

We first define turnstile streaming model (see page 2 of \cite{lnnt16} as an example) as follows
\begin{definition}[Turnstile streaming model]\label{def:turnstile}
    We define two different turnstile streaming models here:
    \begin{itemize}
        \item Strict turnstile: Each update $\Delta$ may be an arbitrary positive or negative number, but we are promised that $x_i \ge 0$ for all $i \in [n]$ at all points in the stream.
        \item General turnstile: Each update $\Delta$ may be an arbitrary positive or negative number, and there is no promise that $x_i \ge 0$ always. Entries in $x$ may be negative.
    \end{itemize}
\end{definition}

Under the turnstile model, the \emph{norm estimation problem} has the following streaming lower bound:

\begin{theorem}[Theorem~1.2 in \cite{bbc+17}]
    \label{thm:low_bound_bbc+17}
    Let $l$ be a symmetric norm on $\R^n$. Any turnstile streaming algorithm (Definition~\ref{def:turnstile})
    that outputs, with probability at least $0.99$, a  $(1\pm1/6)$-approximation for $l(\cdot)$ must use $\Omega(\mmc(l)^2)$ 
    bits of space in the worst case. 
\end{theorem}

We note that this problem is a special case of our symmetry norm distance oracle problem (i.e., with $n=1$ and query vector $q={\bf 0}_d$ where ${\bf 0}_d$ is a all zeros length-$d$ vector). And in this case, the query time of our data structure becomes $\wt{O}(\mmc(l)^2)$ for a constant-approximation, matching the streaming lower bound in Theorem~\ref{thm:low_bound_bbc+17}.  %

\section{Details About Sparse Recovery Tools}
\label{sec:sparse_recv_tool}

In this section we give an instantiation of the sparse recovery tool we use (Definition~\ref{def:bactch_heavy_hitter}) in Algorithms~\ref{alg:batch_heavy_hitter_member_init} - \ref{alg:batch_heavy_hitter_sub_decode}. Although the running times of our data structure are slightly worse (by some log factors) than the result of \cite{knpw11}, it's enough for our symmetric norm estimation task. In terms of space requirement, the classical sparse recovery/compressed sensing only sublinear space is allowed. In our application, we're allowed to use linear space (e.g. $d$ per point, $nd$ in total). And more importantly, our instantiation has much simpler algorithm and analysis than the prior result.

\begin{algorithm}[!ht]\caption{Our CountSketch for Batch heavy hitter}\label{alg:batch_heavy_hitter_member_init}
    \begin{algorithmic}[1]
    \State {\bf data structure} \textsc{BasicBatchHeavyHitter} \Comment{Definition~\ref{def:bactch_heavy_hitter}}
    \State {\bf members} 
        \State \hspace{4mm} $d, n \in \mathbb{N}^+$ \Comment{$n$ is the number of vectors, $d$ is the dimension.}
        \State \hspace{4mm} $\epsilon$ \Comment{We are asking for $\epsilon$-heavy hitters}
        \State \hspace{4mm} $\delta$ \Comment{$\delta$ is the failure probability}
        \State \hspace{4mm} $B$ \Comment{$B$ is the multiple number of each counter to take mean}
        \State \hspace{4mm} $L$ \Comment{$L$ is the number of number of the level of the binary tree.}
        \State \hspace{4mm} $\eta$ \Comment{$\eta$ is the precision for $l_2$ norm estimation.}
        \State \hspace{4mm} $K$ \Comment{$K$ is the number of hash functions.}
        \State \hspace{4mm} $\{h_{l, k}\}_{l \in \{0, \dots, L \} ,k \in [K]} \subseteq [2^l]\times [B]$ \Comment{hash functions.}
        \State \hspace{4mm} $\{C^i_{l, b}\}_{i \in [n], l \in \{0, \dots, L\}, b \in [B]}$ \Comment{The counters we maintain in CountSketch.}
        \State \hspace{4mm} $\{\sigma_{l, k}\}_{l \in \{0, \dots, L\}, k \in [K]} \subseteq [d]\times\{+1, -1\}$ \Comment{The hash function we use for norm estimation}
        \State \hspace{4mm} $\mathcal{Q}$ \Comment{A instantiation of \textsc{LpLpTailEstimation} (Algorithm~\ref{alg:lp_tail_estimation})}
        \State \hspace{4mm} $\{\mathcal{D}_l\}_{l \in [L]}$ \Comment{Instantiations of \textsc{FpEst} (Algorithm~\ref{alg:Fp_estimation})}
    \State {\bf end members}
    \State 
    
    \State {\bf public:}
     \Procedure{Init}{$\epsilon, n \in \mathbb{N}^+, d \in \mathbb{N}^+, \delta$}
        \State $L \gets \log_2 d$
        \State $\delta' \gets \epsilon\delta/(12(\log d) + 1)$
        \For{$l \in \{0, \dots, L\}$}
            \State $\mathcal{D}_l.\textsc{Init}(n, d, l, 1/7, \epsilon, \delta')$ \label{line:fp_est_init}
        \EndFor
        \State $C_0 \gets$ greater than $1000$
        \State $\mathcal{Q}.\textsc{Init}(n, \epsilon^{-2}, 2, C_0, \delta)$ \label{line:tail_est_init}
    \EndProcedure
    \end{algorithmic}
\end{algorithm}
    
\begin{algorithm}[!ht]\caption{Our CountSketch for Batch heavy hitter}\label{alg:batch_heavy_hitter_encode}
    \begin{algorithmic}[1]
    \State {\bf data structure} \textsc{BasicBatchHeavyHitter} \Comment{Definition~\ref{def:bactch_heavy_hitter}}
    \Procedure{EncodeSingle}{$i \in [n], j \in [d], z \in \R, d$}
        \For{$l \in \{0, \dots, L\}$}
            \State $\mathcal{D}_l.\textsc{Update}(i, j, z)$ \label{line:Fp_est_update}
        \EndFor
        \State $\mathcal{Q}.\textsc{Update}(i, j, z)$ \label{line:tail_est_update}
    \EndProcedure
    \State
    
    \Procedure{Encode}{$i \in [n], z \in \R^d, d$}
        \For{$j \in [d]$}
            \State $\textsc{EncodeSingle}(i, j, z_j, d)$\label{line:batch_heavy_hitter_encode_loop}
        \EndFor    
    \EndProcedure
    \State
    \State {\bf end data structure}
    \end{algorithmic}
\end{algorithm}

\begin{algorithm}[!ht]\caption{Our CountSketch for Batch heavy hitter}\label{alg:batch_heavy_hitter_sub_decode}
    \begin{algorithmic}[1]
    \State {\bf data structure} \textsc{BasicBatchHeavyHitter} \Comment{Definition~\ref{def:bactch_heavy_hitter}}
    \Procedure{Subtract}{$i, j, k \in [n]$}
        \For{$l \in \{0, \dots, L\}$}
                \State $\mathcal{D}_l.\textsc{Subtract}(i, j, k)$\label{line:Fp_est_subtract}
        \EndFor
        \State $\mathcal{Q}.\textsc{Subtract}(i, j, k)$\label{line:tail_est_subtract}
    \EndProcedure
    \State
    
    \Procedure{Decode}{$i \in [n], \epsilon, d$}
        \State $\mathrm{EstNorm} \gets \mathcal{Q}.\textsc{Query}(i)$ \Comment{Here the $\mathrm{EstNorm}$ is the estimated tail $l_2$-norm of $i$-vector.}\label{line:tail_est_query}
        \State $S \gets \{0\}, S' \gets \emptyset$
        \For{$l \in \{0, \dots, L\}$} \Comment{The dyadic trick.}
            \For{$\xi \in S$}
                \State $\mathrm{Est} \gets \mathcal{D}.\textsc{Query}(i, \xi)$  \label{line:Fp_est_query}
                \If{$\mathrm{Est} \ge (3/4)\epsilon^2 \cdot \mathrm{EstNorm}$}
                    \State $S' \gets S' \cup \{2\xi, 2\xi + 1\}$
                \EndIf
            \EndFor
            \State $S \gets S'$, $S' \gets \emptyset$
        \EndFor
        \State \Return $S$
    \EndProcedure
    
    \State {\bf end data structure}
    \end{algorithmic}
\end{algorithm}

The following lemma shows that the sparse recovery data structure satisfies our requirements in Theorem~\ref{thm:batch_heavy_hitter}.

\subsection{Our Sparse Recovery Tool}

\begin{algorithm}[!ht]\caption{Batched $\ell_p$ tail estimation algorithm, based on \cite{ns19}}\label{alg:lp_tail_estimation}
    \begin{algorithmic}[1]
    \State {\bf data structure} \textsc{LpLpTailEstimation}
    \State {\bf members}
        \State \hspace{4mm} $m$ \Comment{$m$ is the sketch size}
        \State \hspace{4mm} $\{g_{j,t}\}_{j \in [d], t \in [m]}$ \Comment{random variable that sampled i.i.d. from distribution $\mathcal{D}_p$}
        \State \hspace{4mm} $\{\delta_{j,t}\}_{j \in [d], t \in [m]}$ \Comment{Bernoulli random variable with $\E[\delta_{j, t}] = 1/(100k)$}
        \State \hspace{4mm} $\{y_{i,t}\}_{i \in [n], t \in [m]}$ \Comment{Counters}
    \State {\bf end members}
    
    \State {\bf public:}
    \Procedure{Init}{$n, k, p, C_0, \delta$}
        \State $m \gets O(\log(n/\delta))$
        \State Choose $\{g_{j,t}\}_{j \in [d], t \in [m]}$ to be random variable that sampled i.i.d. from distribution $\mathcal{D}_p$
        \State Choose $\{\delta_{j,t}\}_{j \in [d], t \in [m]}$ to be Bernoulli random variable with $\E[\delta_{j, t}] = 1/(100k)$
        \Comment{Matrix $A$ in Lemma~\ref{lem:lp_tail_estimation} is implicitly constructed based on $g_{j,t}$ and $\delta_{j,t}$}
        \State initialize $\{y_{i,t}\}_{i \in [n], t \in [m]} =\{0\}$ 
    \EndProcedure
    \State
    
    \Procedure{Update}{$i \in [n], j \in [d], z \in \R$}
        \For{$t \in [m]$}
            \State $y_{i, t} \gets y_{i, t} + \delta_{j, t} \cdot g_{j t} \cdot z$
        \EndFor
    \EndProcedure
    \State
    
    \Procedure{Subtract}{$i, j, k \in [n]$}
        \For{$t \in [m]$}
            \State $y_{i,t} \gets y_{j,t} - y_{k,t}$
        \EndFor
    \EndProcedure
    \State
    
    \Procedure{Query}{$i \in [n]$}
        \State $V \gets \median_{t \in [m]}|y_{i,t}|^2$
        \State \Return $V$
    \EndProcedure

    \end{algorithmic}
\end{algorithm}

We first state the correctness, the proof follows from framework of \cite{knpw11}.
\begin{lemma}
    The function \textsc{Decode}$(i, \epsilon, d, \delta)$ (Algorithm~\ref{alg:batch_heavy_hitter_sub_decode}) returns a set $S \subseteq d$ of size $|S| = O(\epsilon^{-2})$ containing all $\epsilon$-heavy hitters of the $i$-column of the matrix under $l_2$ with probability of $1-\delta$. Here we say $j$ is an $\epsilon$-heavy hitter under $l_2$ if $|x_j| \ge \epsilon \cdot \|x_{\overline{[\epsilon^{-2}]}}\|_2$ where $x_{\overline{[k]}}$ denotes the vector $x$ with the largest $k$ entries (in absolute value) set to zero. Note that the number of heavy hitters never exceeds $2/\epsilon^2$.
\end{lemma}

\begin{proof}

    The correctness follows from the framework of \cite{knpw11}, and combining tail estimation (Lemma~\ref{lem:lp_tail_estimation}) and norm estimation.
    
\end{proof}

We next analyze the running time of our data structure in the following lemma:

\begin{lemma}
    The time complexity of our data structure (Algorithm~\ref{alg:batch_heavy_hitter_member_init}, Algorithm~\ref{alg:batch_heavy_hitter_encode} and Algorithm~\ref{alg:batch_heavy_hitter_sub_decode}) is as follows:
    \begin{itemize}
        \item \textsc{Init} takes time of $O(\epsilon^{-1} (n + d) \log^2(  n d /  \delta  ) )$.
        \item \textsc{EncodeSingle} takes time of $O(\log^2(  n d / \delta  ) )$.
        \item \textsc{Encode} takes time of $O(d \log^2(  n d / \delta  ) )$.
        \item \textsc{Subtract} takes time of $O(\epsilon^{-1} \log^2(  n d  / \delta ) )$.
        \item \textsc{Decode} takes time of $O(\epsilon^{-2} \log^2(  n d / \delta ) )$.
    \end{itemize}
\end{lemma}

\begin{proof}

    We first notice that $L = \log_2 d$ and $\delta' = \epsilon \delta /(12(\log( d) + 1))$.
    
    For the procedure \textsc{Init} (Algorithm~\ref{alg:batch_heavy_hitter_member_init}), Line~\ref{line:fp_est_init} takes time 
    \begin{align*}
        O(L \epsilon^{-1} n \log (n /\delta')) = O(\epsilon^{-1} n \log^2(\delta^{-1} \epsilon^{-1} n d \log( d )).
    \end{align*}
    Line~\ref{line:tail_est_init} takes time $O((n + d)\log(n/\delta))$. Taking together, we have the total running time of \textsc{Init} is $O(\epsilon^{-1} (n + d) \log^2(\delta^{-1} \epsilon^{-1} n d \log( d ))$
    
    For the procedure \textsc{EncodeSingle} (Algorithm~\ref{alg:batch_heavy_hitter_encode}), Line~\ref{line:Fp_est_update} takes time of 
    \begin{align*}
        O(L \log(n/\delta')) = O(\log^2(\epsilon^{-1}\delta^{-1} n d \log(d))).
    \end{align*}
    Line~\ref{line:tail_est_update} takes time of $O(\log(n / \delta))$. Taking together we have the total running time of \textsc{EncodeSingle} to be $O(\log^2(\epsilon^{-1}\delta^{-1} n d \log(d)))$.
    
    For the procedure \textsc{Encode} (Algorithm~\ref{alg:batch_heavy_hitter_encode}), it runs \textsc{EncodeSingle} for $d$ times, so its running time is $O(d \log^2(\epsilon^{-1}\delta^{-1} n d \log(d)))$.
    
    For the procedure \textsc{Subtract} (Algorithm~\ref{alg:batch_heavy_hitter_sub_decode}), Line~\ref{line:Fp_est_subtract} runs in time 
    \begin{align*}
        O(L \epsilon^{-1} \log(n/\delta')) = O(\epsilon^{-1} \log^2(\epsilon^{-1}\delta^{-1} n d \log(d))).
    \end{align*}
    Line~\ref{line:tail_est_subtract} runs in time $O(\log(n/\delta))$. So the total running time is $O(\epsilon^{-1} \log^2(\epsilon^{-1}\delta^{-1} n d \log(d)))$.
    
    For the procedure \textsc{Decode} (Algorithm~\ref{alg:batch_heavy_hitter_sub_decode}), Line~\ref{line:tail_est_query} runs in time $O(\log(n/\delta))$. Line~\ref{line:Fp_est_query} runs in time
    \begin{align*}
        O(L \epsilon^{-2} \log(n/\delta')) = O(\epsilon^{-2} \log^2(\epsilon^{-1}\delta^{-1} n d \log(d))).
    \end{align*}
    So the total running time is $O(\epsilon^{-2} \log^2(\epsilon^{-1}\delta^{-1} n d \log(d)))$
    
    Thus we complete the proof.
\end{proof}

The space complexity of our data structure is stated in below.

\begin{lemma}
    Our batch heavy hitter data structure (Algorithm~\ref{alg:batch_heavy_hitter_member_init}, Algorithm~\ref{alg:batch_heavy_hitter_encode} and Algorithm~\ref{alg:batch_heavy_hitter_sub_decode}) takes the space of
    \begin{align*}
        O(\epsilon^{-1} (n + d) \log^2(\epsilon^{-1}\delta^{-1} n d \log(d))).
    \end{align*}
\end{lemma}

\begin{proof}
    Our data structure has these two parts to be considered:
    \begin{itemize}
        \item The instantiations of \textsc{FpEst} we maintain.
        \item The instantiation of \textsc{LpLpTailEstimation} we maintain .
    \end{itemize}
    The first part takes the space of 
    \begin{align*}
        O(L \epsilon^{-1} n \log(n/\delta')) = O(\epsilon^{-1} n \log^2(\epsilon^{-1}\delta^{-1} n d \log(d))).
    \end{align*}
    And the second part takes the space of $O(d\log(n/\delta))$. Adding them together we complete the proof.
\end{proof}

\subsection{Lp Tail Estimation}

\cite{ns19} provide a linear data structure \textsc{LpLpTailEstimation}$(x,k,p,C_0,\delta)$ that can output the estimation of the contribution of non-heavy-hitter entries. We restate their Lemma as followed.

\begin{lemma}[Lemma~C.4 of \cite{ns19}]
\label{lem:lp_tail_estimation}
    Let $C_0 \ge 1000$ denote some fixed constant. There is an oblivious construction of matrix $A \in \R^{m\times n}$ with $m = O(\log(1/\delta))$ along with a decoding procedure \textsc{LpLpTailEstimation}$(x,k,p,C_0,\delta)$ such that, given $Ax$, it is possible to output a value $V$ in time O(m) such that
    \begin{align*}
        \frac{1}{10k}\|x_{\overline{C_0k}}\|_p^p \le V \le \frac{1}{k}\|x_{\overline{k}}\|_p^p,
    \end{align*}
    holds with probability $1 - \delta$.
\end{lemma}

\begin{lemma}\label{lem:lp_tail_estimation_time}
The running time of the above data structure is 
\begin{itemize}
    \item \textsc{Init} runs in time of $O((n + d) \log(n/\delta))$
    \item \textsc{Update} runs in time of $O(\log(n/\delta))$
    \item \textsc{Subtract} runs in time of $O(\log(n/\delta))$
    \item \textsc{Query} runs in time of $O(\log(n/\delta))$
\end{itemize}
And its space is $O(d \log(n / \delta))$.
\end{lemma}

\subsection{Lp Norm Estimation}

\begin{algorithm}[!ht]\caption{Batched $\ell_p$ norm estimation algorithm, based on \cite{knw10}}\label{alg:Fp_estimation}
    \begin{algorithmic}[1]
    \State {\bf data structure} \textsc{LpNormEst}
    \State {\bf members} 
        \State \hspace{4mm} $R, T$ \Comment{parallel parameters}
        \State \hspace{4mm} $m$ \Comment{Sketch size}
        \State \hspace{4mm} $\{y^i_{r, t}\}_{i \in [n], r \in [R], t \in [T]} \subset \R^{m}$ \Comment{Sketch vectors}
        \State \hspace{4mm} $\{A\}_{r, t} \subset \R^{m \times 2^l}$ \Comment{Sketch matrices}
        \State \hspace{4mm} $\{h_t\}_{t \in [T]}:\{0, \dots, 2^l\} \rightarrow [R]$ \Comment{hash functions}
    \State {\bf end members} 
    \State
    
    \State {\bf public:}
    \Procedure{Init}{$n, d, l, \phi, \epsilon$}
        \State $R \gets \lceil 1/\epsilon \rceil$
        \State $T \gets \Theta(\log(n/\delta))$
        \State $m \gets O(1/\phi^2)$ 
        \For{$i \in [n], r \in [R], t \in [T]$}
            \State $y^i_{r, t} \gets \textbf{0}$ \Comment{Sketch vectors}
        \EndFor
        
        \For{$i \in [n], r \in [R], t \in [T]$}
            \State generate $A_{r, t} \in \R^{m \times 2^l}$ \Comment{See \cite{knw10} for details}
        \EndFor
        \State initialize $h$
    \EndProcedure
    \State
    
    \Procedure{Update}{$i \in [n], j \in [d], z \in \R$}
        \State $\xi \gets  j \cdot 2^l / d $
        \For{$t \in [T]$}
            \State $y^i_{h_t(\xi), t} \gets y^i_{h_t(\xi), t} + A_{h_t(\xi), t}\mathbf{i}^l_{j, z}$ \Comment{$\mathbf{i}^l_{j, z}$ is define to be the $2^l$-dimensional vector with $j$-th entry of $z$ and others to be $0$}
        \EndFor
    \EndProcedure
    \State
    
    \Procedure{Subtract}{$i, j, k \in [n]$}
        \For{$r \in [R], t \in [T]$}
            \State $y^i_{r, t} \gets y^j_{r, t} - y^k_{r, t}$
        \EndFor
    \EndProcedure
    \State
    
    \Procedure{Query}{$i \in [n], \xi \in [2^l]$}
        \For{$t \in [T]$}
            \State $V_{h_t(\xi), t} \gets \median_{\zeta \in [m]}y^i_{r,t,\zeta}/\median(|\mathcal{D}_p|)$ \Comment{Definition~\ref{def:p_stable_distribution}}
        \EndFor
        \State $V \gets \median_{t \in [T]}V_{h_t(\xi), t}$
        \State \Return $V$
    \EndProcedure

    \end{algorithmic}
\end{algorithm}

Following the work of \cite{knw10}, we provide a linear sketch satisfying the following requirements.

\begin{lemma}[\cite{knw10}]
\label{lem:linear_sketch_for_each_level}
There is a linear sketch data structure \textsc{FpEst} using space of 
\begin{align*}
O(\epsilon^{-1} \phi^{-2} n \log(n/\delta))
\end{align*}
and it provide these functions:
\begin{itemize}
    \item \textsc{Init}$(n \in \Z^+, d \in \Z^+, l \in \Z^+, \phi, \epsilon, \delta)$: Initialize the sketches, running in time $O(\epsilon^{-1} n \log(n/\delta))$
    \item \textsc{Update}$(i \in [n], j \in [d], z \in \R)$: Update the sketches, running in time $O(\log(n / \delta))$
    \item \textsc{Subtract}$(i, j, k \in [n])$: Subtract the sketches, running in time $O(\epsilon^{-1} \phi^{-2} \log(n/\delta))$
    \item \textsc{Query}$(i \in [n], \xi \in [2^l])$: This function will output a $V$ satisfying
    \begin{align*}
        (1 - \phi) \cdot F_i(l, \xi) \le V \le (1 + \phi) \cdot (F_i(l, \xi) + 5\epsilon \|x_i\|_2^2),
    \end{align*}
    where $F(l,\xi)$ is defined as
    \begin{align*}
        F_i(l, \xi) := \sum_{j = \frac{\xi}{2^l}d}^{\frac{\xi + 1}{2^l}d - 1} |x_{i, j}|^2,
    \end{align*} running in time $O(\phi^{-2}\log(n/\delta))$
\end{itemize}

\end{lemma}
We choose $\phi$ to be constant when we use the above Lemma.

\end{document}